\documentclass[10pt, journal, onecolumn]{IEEEtran}

\usepackage{subfig}

\usepackage{epsfig,amsmath,amssymb,epsf,cite,algorithm,algorithmic}
\usepackage{graphicx}
\usepackage{epstopdf}
\usepackage[margin=1in]{geometry}

\newtheorem{theorem}{Theorem}
\newtheorem{lemma}{Lemma}

\newtheorem{proposition}{Proposition}
\newtheorem{definition}{Definition}

\usepackage{stmaryrd}
\usepackage{amssymb}
\usepackage{tipa}
\usepackage{setspace}
\usepackage[none]{hyphenat}
\usepackage{color}

\graphicspath{{./figs/}}

\def\bLambda{{\pmb{\Lambda}}}

\def\b0{{\pmb{0}}}

   \def\bx{{\boldsymbol{x}}}
  \def\bs{{\boldsymbol{s}}} \def\by{{\boldsymbol{y}}}
   
\def\bg{{\boldsymbol{g}}} 

\def\bA{{\boldsymbol{A}}}   \def\bU{{\boldsymbol{U}}}
 \def\bI{{\boldsymbol{I}}}  \def\bV{{\boldsymbol{V}}}
  \def\bQ{{\boldsymbol{Q}}} \def\bW{{\boldsymbol{W}}}
   
  \def\bS{{\boldsymbol{S}}} \def\bY{{\boldsymbol{Y}}}
   \def\bZ{{\boldsymbol{Z}}}
\def\bG{{\boldsymbol{G}}}

\def\bgg{{\underline{\boldsymbol{g}}}}

\newenvironment{proof}[1][Proof:]{\begin{trivlist}
\item[\hskip \labelsep {\bfseries #1}]}{\end{trivlist}}

\newcommand{\qed}{\hfill \mbox{\raggedright \rule{.07in}{.1in}}}
\newcommand{\DG}[1]{{#1}}
\definecolor{ForestGreen}{RGB}{34,139,34}
\newcommand{\LL}[1]{{#1}}
\definecolor{orange}{RGB}{255,165,0}
\newcommand{\Lina}[1]{{#1}}
\newcommand{\XC}[1]{{#1}}

\begin{document}
%
\title{\DG{Asynchronous Massive Access and Neighbor Discovery Using
OFDMA}}

\author{
\IEEEauthorblockN{Xu~Chen\IEEEauthorrefmark{1},
\DG{Lina~Liu\IEEEauthorrefmark{2},}
Dongning~Guo\IEEEauthorrefmark{2}, Gregory W. Wornell\IEEEauthorrefmark{3}}
\\\IEEEauthorblockA{\IEEEauthorrefmark{1} Waymo, Mountain View, CA}
    \\\IEEEauthorblockA{\IEEEauthorrefmark{2}Department of Electrical and Computer Engineering, Northwestern University, Evanston, IL}
    \\\IEEEauthorblockA{\IEEEauthorrefmark{3}Department of Electrical Engineering and Computer Science, Massachusetts Institute of Technology, Cambridge, MA}
}

\date{}

\maketitle

\begin{abstract}
\DG{\LL{The fundamental communication problem in the wireless Internet of Things (IoT) is to discover a massive number of devices and to allow them reliable access to shared channels. Oftentimes these devices} transmit short messages randomly and sporadically.  This paper proposes a novel signaling scheme for grant-free massive access, where each device encodes its \LL{identity} and/or \LL{information} in a sparse set of tones.  Such transmissions are implemented in the form of orthogonal frequency-division multiple access (OFDMA).
\LL{Under some mild conditions and} assuming device delays to be bounded \LL{unknown} multiples of symbol intervals, sparse OFDMA is proved to enable arbitrarily reliable asynchronous device identification and message decoding with a codelength that is \LL{$O(K(\log K+\log S + \log N))$, where $N$ denotes the device population, $K$ denotes 
the actual number of active devices, and $\log S$ is essentially equal to the number of bits a device can send (including its identity).}
}
By exploiting \DG{the} \LL{Fast Fourier Transform (FFT)}, the computational complexity \LL{for discovery and decoding} 
\DG{can be made to be} sub-linear in the total device population.
To prove the concept, a specific design is proposed to identify up to 100 active devices out of $2^{38}$ possible devices with up to 20 symbols of delay and moderate signal-to-noise ratios and fading.  The codelength compares much more favorably with those of standard slotted ALOHA and carrier-sensing multiple access (CSMA) schemes.
%
\end{abstract}

\section{Introduction}

By some estimate~\cite{stoyanova2020survey}, there will be up to 500 billion connected devices world-wide in the Internet of Things (IoT) by year 2030. There can be well over a million \DG{low-cost, battery-powered} IoT devices within 500 meter range in a densely populated area. 
\DG{The general term of {\em massive access} describes the setting where a large number of devices need to access a shared medium in the uplink to send messages to some access points.  Usually, wireless IoT devices transmit short messages randomly and sporadically.  If the access points only need to decode the messages but are not interested in the corresponding devices' identities, the setting is referred to as {\em unsourced random access}~\cite{polyanskiy2017perspective}.  In the other extreme where the sole purpose is to detect and identify active devices within range, the setting is often referred to as {\em neighbor discovery} or {\em device identification}.  Whereas neighbor discovery is a special case of massive access, the latter can also be regarded as the former if the transmitted data are regarded as all or a segment of the device identities.
}


\DG{Most IoT access solutions in practice either orthogonalize transmissions or suffer from collisions over the air.
In particular, using} 
a naive time division multiple access (TDMA) scheme to schedule \DG{densely deployed} 
devices would incur \DG{large latencies}.  
\DG{So would a} random access mechanism \DG{based on} 
classical ALOHA. 
\DG{Since} 
the number of devices far exceeds the frame length, it is also impossible to assign nearly orthogonal sequences to all the devices to support code division multiple access (CDMA), especially due to device asynchrony.
Moreover, massive access using conventional multiuser detection approaches
generally involves a polynomial complexity in the number of devices~\cite{angelosante2010neighbor}, which is also unaffordable.  

\DG{In practice, it is hard to eliminate relative delays between devices, in part because of their different distances to the same access point.}
For example, a \DG{difference} 
of 300 meters implies \DG{a} free space propagation delay of one microsecond, which spans 20 samples at 20 million samples per second.
\DG{On the other hand, the maximum \LL{relative} delay can be within a small fraction of a frame if a common timing reference such as a beacon signal is available.}

\DG{A successful massive access solution for the IoT should be {\em ultra-scalable} to support a massive number of devices, incur low latencies, and allow for some asynchrony between devices.}

\subsection{Related Work}
\label{sec::related_work}

\DG{To detect and/or decode messages from many transmitters sharing a common medium is a problem in the well established area of multiuser detection.  
A small-scale neighbor discovery 
problem is studied} as a multiuser detection problem \DG{in}~\cite{angelosante2010neighbor}.
\DG{In a more challenging case where the device population is several orders of magnitude larger than} 
the number of active devices, 
schemes inspired by \DG{the} compressed sensing \LL{(CS)} \DG{literature are} 
proposed \DG{in}~\cite{zhu2011exploiting,schepker2012compressive,zhang2013neighbor,zhang2014virtual,liu2018massive1,thompson2018compressed,chen2018sparse,liu2018sparse,amalladinne2018asynchronous}. These schemes can reduce the codelength 
\DG{substantially when} compared with 
802.11 type protocols, but most of those schemes require synchronous transmissions. 
\DG{A \LL{convex optimization} based algorithm is} 
proposed to detect active devices 
\DG{using} asynchronous CDMA random access~\cite{applebaum2012asynchronous}, but
\LL{the} \DG{polynomial complexity \LL{in ~\cite{zhu2011exploiting,schepker2012compressive,zhang2013neighbor,zhang2014virtual,liu2018massive1,thompson2018compressed,chen2018sparse,liu2018sparse,amalladinne2018asynchronous,applebaum2012asynchronous}} scales poorly for massive access}.

Non-orthogonal multiple access (NOMA) allows multiple devices to share the time and frequency resources via power domain or code domain multiplexing~\cite{saito2013non, dai2015non}, \LL{where successive interference cancellation (SIC)} is used to cancel multiuser interference at the receiver. NOMA's decoding algorithms (e.g., message passing~\cite{taherzadeh2014scma}) have in general polynomial complexities in the \DG{device population}. However, whether NOMA is resilient to imperfect channel state information is not well understood. 

\DG{The} information-theoretic limits of \DG{uplink and downlink} massive access \DG{are}
studied in~\cite{chen2017capacity} and~\cite{chen2014manybroadcast}, \DG{respectively}, where the number of devices 
scales with the codelength \DG{in general}. 
\DG{It is shown therein} that separate device identification and message decoding is capacity-achieving.  The degree of freedom of massive access fading channels \DG{is} 
analyzed in~\cite{yu2017fundamental}. The capacity of asynchronous massive access \DG{is} 
analyzed in~\cite{shahi2018strongly}. \LL{SIC does not achieve the boundary of the capacity region in the case of asynchronous NOMA~\cite{ganji2021asynchronous}.} An alternative information-theoretical model \LL{is} proposed in~\cite{polyanskiy2017perspective}, where the access point aims to decode the set of messages regardless of the corresponding sources.
\LL{Low-complexity schemes have been proposed \LL{for unsourced multiple access}~\cite{ordentlich2017low,amalladinne2018coupled,fengler2019sparcs,calderbank2018chirrup}. Reference \cite{ordentlich2017low} proposes $T$-fold ALOHA. Reference \cite{amalladinne2018coupled} uses CS for sub-block recovery and a tree-based algorithm for sub-block stitching. A sparse regression code is used in \cite{fengler2019sparcs} as an inner code with a corresponding approximate message passing based algorithm as a inner decoder. Binary chirp coding and chirp reconstruction algorithm is studied in \cite{calderbank2018chirrup}. These schemes require fully synchronous transmissions.}


The idea of codes on graph has \DG{also} been applied in random access~\cite{paolini2015coded,taghavi2016design}.
One \DG{notable} scheme is 
{\em coded slotted ALOHA, in which} 
packets are repeatedly transmitted in different slots and are decoded using successive cancellation \DG{under the assumptions of} 
synchronous transmissions and perfect 
cancellation. The asynchronous model \DG{is} 
studied in~\cite{de2014asynchronous,clazzer2016detection,sandgren2016frame}
\DG{under the idealized assumption of perfect cancellation}.
Rateless codes have been proposed for multiple access in machine-to-machine communications~\cite{shirvanimoghaddam2015probabilistic}, where the channel gains are assumed to be known. As the number of users increases, the imperfect channel estimation is detrimental to the performance of successive cancellation.

\subsection{Contributions}

\LL{The goal here is to develop a suitable signaling scheme for grant-free transmissions by a large number of devices with moderate delay uncertanties.} \DG{The \LL{specific} contributions of this paper are summarized as follows:}
\begin{enumerate}
\item We propose a novel scheme, referred to as sparse OFDMA, \DG{where a device encodes its information and/or identity \LL{onto} a sparse set of orthogonal tones.} 
We exploit the fact that the tones' frequencies are invariant to delays. Sparse OFDMA 
is highly effective in an asynchronous setting, which is particularly appealing in the context of the IoT.

\item \DG{The proposed scheme} 
utilizes low-complexity point-to-point capacity-approaching codes 
\DG{and} exploits multiuser diversity \DG{using} 
successive cancellation. 

\item  
\DG{In the asymptotic regime where}
\LL{
the number of active devices and the maximum delay are sublinear relative to the device population, sparse OFDMA correctly detect all active devices with high probability, using an essentially sublinear codelength in the device population. With a bounded maximum delay, under mild conditions for the number of active devices, the computational complexity is also sublinear in the device population (due to a sparse Fourier transform).}



\end{enumerate}

\subsection{Paper Organization and Notations}

The rest of the paper is organized as follows. Section~\ref{sec:system} presents the system model and main results. Section~\ref{sec:signaling} describes the signalling scheme of sparse OFDMA. Section~\ref{sec:async} presents the asynchronous massive access algorithm. Section~\ref{sec:proof_sync} and Section~\ref{sec:proof_asyn} prove theoretical performance guarantees for synchronous and asynchronous transmissions, respectively. Section~\ref{sec:sim} presents some numerical results on a practical design. Section~\ref{sec:conclude} concludes the paper.

For ease of notation, the index of a vector or each dimension of a matrix starts from 0 throughout the paper. The elements of a $B \times C$ matrix are denoted as $y_b^c$, where $c = 0, \cdots, C-1$ and $b = 0, \cdots, B-1$. We write the $b$-th row vector as $\by_b = \left( y^0_b, \cdots, y^{C-1}_b \right)$ and the $c$-th column vector as $\by^c = \left( y^c_0, \cdots, y^c_{B-1} \right)^T$. We denote the real and imaginary parts of a variable $X$ as $X_R$ and $X_I$, respectively. All logarithms are base 2.

\section{System Model and Main Results}
\label{sec:system}

We focus on the uplink access problem in a system consisting of one access point and $N$ potential devices in total.
 Let $\{0, \cdots, N-1\}$ denote the set of device indices. An arbitrary subset of devices are \LL{within the range of the access point} and active, whose indices form the set $\mathcal{K} \subseteq \{0, \cdots, N-1\}$.  Let $K = |\mathcal{K}|$ denote the number of active devices. For tractability, we assume symbol synchrony without frame synchrony, i.e., the delay of each device's signal at the access point is an integer number of symbol intervals. We further assume the delays are similar. This can be accomplished by using a common beacon to trigger transmissions. 
Specifically, let the delay of any device relative to a reference at the access point be \DG{no more than} 
$M$ symbol intervals.  Fig.~\ref{fig:system} illustrates a small example with three devices.

Let $\bs_k = (s_{k,0},\dots,s_{k,L-1} )$ denote the $L$-symbol codeword transmitted by device $k$.
In the absence of frequency selectivity, the received signal at every (integer) time $i$ is given by
\begin{align}\label{eq:system_model}
x_i = \sum\limits_{k \in \mathcal{K}} a_{k} s_{k ,i-m_{k}} + w_i,
\end{align}
where $a_{k} \in \mathbb{C}$ is the channel coefficient, $m_{k}$ is the transmission delay of device $k$, and $w_i \sim \mathcal{C N} (0, 2 \sigma^2)$ are independently and identically distributed (i.i.d.) circularly symmetric complex Gaussian random variables. The discovery scheme is based on signals within a single codeword duration. From each receiver's point of view, the signal $s_{k,i} = 0$ for $i < 0$ and $i \geq L$ for all $k$. 


\begin{figure}
  \centering
  \includegraphics[width=9cm]{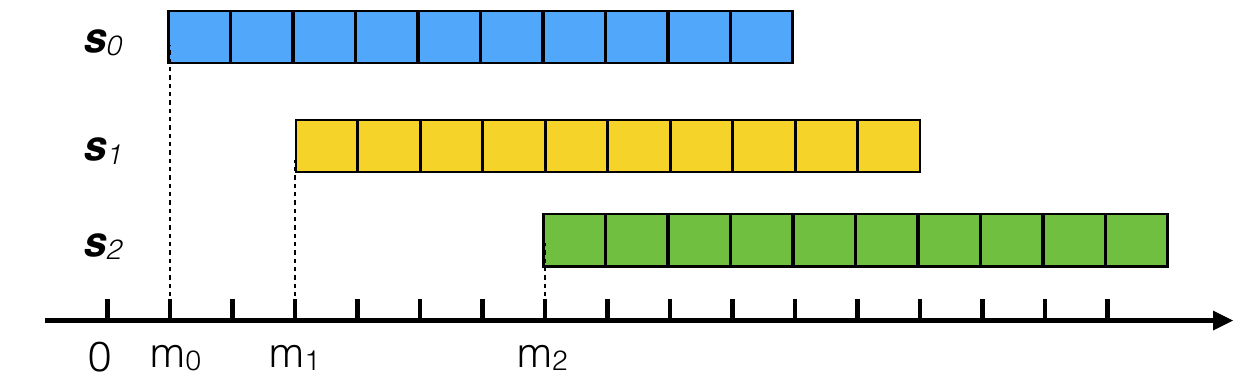}\\
  \caption{Frame-asynchronous symbol-synchronous three-user model.}\label{fig:system}
\end{figure}

We shall design a \LL{massive access} scheme with relatively small codelength and low computational complexity. 
The following two theorems are \LL{our} key results 
\DG{for the synchronous case and the asynchronous case, respectively}:

\begin{theorem}\label{theorem:noisy_discovery_sync}
Suppose 
each device's message set contains no more than $S$ messages.
Suppose, out of $N$ devices, an unknown subset of devices transmit. 
Suppose 
\DG{their signals} are perfectly aligned \DG{at the receiver, so that}
the maximum delay $M=0$.
Suppose \DG{also} the noise variance is fixed and the received signal amplitude of every active device is at least $a$. Then for every $a, \epsilon >0$, there exist $\alpha_0, \alpha_1, K_0 >0$ such that for every $N$ and $K$ satisfying \Lina{$N \geq K \geq K_0$}, there exists a code of length
\begin{align} \label{eq:Lsync}
    L \leq \alpha_0 K (\log N + \log S + \log K)
\end{align}
such that as long as no more than $K$ devices transmit, all their identities and messages will be decoded correctly with probability no less than $1-\epsilon$. Moreover, this can be accomplished using fewer than \Lina{$ \alpha_1 K (\log K) ( \log N + \log S + \log K)$} arithmetic operations.
\end{theorem}

\begin{theorem}\label{theorem:noisy_discovery_asyn}
Suppose 
each device's message set contains no more than $S$ messages.
Suppose, out of $N$ devices, an unknown subset of devices transmit. 
Suppose the 
delay of every active device is an integer number of symbol intervals less than $M \geq 0$. Suppose \DG{also} the noise variance is fixed and the received signal amplitude of every active device is bounded between $a$ and $\bar{a}$. Then for every $a, \bar{a}, \epsilon >0$ \DG{with} 
$a \leq \bar{a}$, there exist $\alpha_0, \alpha_1, K_0 >0$ such that for every $N$ and $K$ satisfying \Lina{$N \geq K \geq K_0$},
there exists a code of length
\begin{align} \label{eq:Lasync}
    L \leq \alpha_0 \left( (K+M) (\log N  + \log S + \log K) + K \log(M+1) \right)
\end{align}
such that as long as no more than $K$ devices transmit, all their identities and messages will be decoded correctly with probability no less than $1-\epsilon$. Moreover, this can be accomplished using fewer than \Lina{$ \alpha_1 ( K (\log K) (\log N + \log S + \log K) + K M^2 +  K^2 M  \log (KM+1) ) $} arithmetic operations.
\end{theorem}

Theorems~\ref{theorem:noisy_discovery_sync} \DG{and~\ref{theorem:noisy_discovery_asyn} provide} 
a scaling law for the complexity and codelength in terms of the device population,
the number of active devices,
\DG{and the delay bound}. Synchronous massive access, i.e., $M = 0$, requires a smaller codelength and fewer arithmetic operations than the asynchronous case. Theorem~\ref{theorem:noisy_discovery_asyn} reduces to Theorem~\ref{theorem:noisy_discovery_sync} in the spectial case of $M=0$. 


\DG{With fixed maximum delay $M$,} \LL{when the codebook size $S=O(N)$ and} the number of active users $K$ satisfies \LL{$K\log K = o(N)$}, the codelengths for both the synchronous and asynchronous schemes are sublinear in the \LL{device population $N$ according to Theorems \ref{theorem:noisy_discovery_sync} and \ref{theorem:noisy_discovery_asyn}}.
\LL{Furthermore, with fixed $M$, if $K^2\log K = o (N)$}, the number of arithmetic operations involved in the asynchronous scheme is also sublinear in $N$. 





\section{Sparse OFDMA Signaling}
\label{sec:signaling}

In this section, we propose a specific signaling 
scheme. Let the spectrum be divided into $B$ orthogonal subcarriers, \LL{where $B$ is much smaller than the device population $N$. It is not possible to design mutually orthogonal OFDM symbols for all $N$ devices. We let each active device transmit several OFDM symbols on a sparse subset of the subcarriers}. The scheme is thus referred to as sparse OFDMA. 
We shall show that sparse OFDMA of moderate symbol duration can accomodate \LL{a large number of devices with bursty transmissions}.

We develop sparse OFDMA in four steps. First, we consider noiseless device identification \LL{with fewer devices} than the number of subcarriers \LL{($N \le B$)}. Second, we consider noiseless device identification where $N > B$ and a single device is active. Third, we consider the general noisy device identification, where $N > B$ and $K \ll N$ devices are active. At last, we consider 
simultaneous device identification and message decoding.

\subsection{Noiseless Device Identification \LL{with $N \le B$}}

The key idea for addressing arbitrary \LL{delays} is to use the fact that the frequency of a sinusoidal signal is invariant to delay, where the delay merely causes a phase shift. Since the delay is bounded by $M$, we include $M$ samples in the OFDM symbol as a cyclic prefix. Hence, each OFDM symbol contains $B+M$ samples. Since $N \leq B$, \LL{device $k$ can be assigned the unique subcarrier $k$}. The \LL{transmitted} discrete-time signal structure is given by
\begin{align}\label{eq:codeword}
s_{k,i} =   g_{k} \exp \left( \frac{\iota 2 \pi k i}{B} \right), \quad i=0, \cdots, B+M-1,
\end{align}
where $g_{k} \in \mathbb{R}$ is a known design parameter of unit amplitude and $\iota^2 = -1$.

At the receiver side, the signals from all the neighbors arrive after a reference frame start point. The receiver discards the first $M$ samples of each sparse OFDMA symbol and collect the next $B$ samples as $\by = (y_0, \cdots, y_{B-1} )$, where $ y_i = x_{i+M} $, $i =0, \cdots, B-1$. If each device is assigned a unique subcarrier, performing $B$-point discrete Fourier transform (DFT) on $\by$ yields nonzero at the $k$-th subcarrier if and only if device $k$ is active. The delay $m_k$ only affects the phase of the DFT value. Therefore, the signaling scheme~\eqref{eq:codeword} is sufficient to detect the active devices in a noiseless case with computational complexity of $O(B \log B)$ needed by the Fast Fourier Transform (FFT) algorithm.

\subsection{Noiseless Single Device Identification \LL{with $N > B$}}
\label{sec:device_identify_b}

\LL{With $N > B$}, it is impossible to assign a distinct subcarrier to each device. Suppose we randomly assign subcarrier $b_k \in \{0, \cdots, B-1\}$ to device $k$ and still apply the signaling scheme give by \eqref{eq:codeword}. If performing $B$-point DFT on $\by$ yields nonzero at subcarrier $b$, we cannot \LL{be sure about} which device is active if multiple devices are assigned to this subcarrier. We resolve this ambiguity \DG{by including several OFDM symbols in a frame and embedding} 
the device's identity through coefficient $g_{k}$ in~\eqref{eq:codeword} across \LL{the} OFDM symbols for transmission.

Let \LL{$J = \lceil \log N \rceil$ and let $(k)_2 = (k_1, \cdots, {k}_{J})$} denote the binary representation of device index $k$. We \LL{simply adopt the design of $\left( g^0_{k}, \cdots, g^{J}_{k} \right) = \left( 1, (-1)^{k_1}, \cdots, (-1)^{k_{J} } \right)$}. We let device $k$ transmit \LL{$\left( \bs^0_{k}, \cdots, \bs^{J}_{k} \right)$}, where \LL{$\bs^j_{k} = (s^j_{k,0}, \cdots, s^j_{k,B+M-1})$ and
\begin{align}
s^j_{k,i} =   g^j_{k} \exp \left( \frac{\iota 2 \pi b_{k} i}{B} \right),
\end{align}
for $i = 0, \cdots, B+M-1$ and $j = 0, \cdots, J$.} The code length is thus \LL{$(J+1)(B+M)$} samples. For the \LL{$j$-th} OFDM symbol, we discard the first $M$ samples and use the next $B$ samples to form a vector \LL{$\by^j = (y^j_0, \cdots, y^j_{B-1})$}, where 
\LL{\begin{align}
y^j_i &= x_{i+j(B+M)+M} \\
&=a_k s_{k, i+M-m_k}^j,
\end{align}}
$i =0, \cdots, B-1$. Performing $B$-point DFT on \LL{$\by^j$} yields
\LL{
\begin{align}\label{eq:Yk_noiseless}
Y^j_b &= \frac{1}{B} \sum_{i=0}^{B-1}  \exp \left( - \frac{ \iota 2 \pi b i}{B} \right) a_k s_{k, i+M-m_k}^j \\
& = \left\{
\begin{array}{cc}
0, &\text{ if } b \neq b_k \\
A_{k,b}g_k^j,  &\text{ if } b=b_k,
\end{array} \right.
\end{align}
where
\begin{equation}\label{eq:Akb}
A_{k,b}=a_k\exp\left(\frac{\iota 2\pi b_k(M-m_k)}{B}\right).
\end{equation}}As in the $B \geq N$ case, the delay $m_k$ only affects the phase of the received signal from device $k$.

Subcarrier $b$ is associated with a \LL{length-$(J+1)$ vector $\bY_b = (Y_b^0, \cdots, Y_b^{J})$.} It can be seen that $g^0_{k}$ serves as a reference symbol capturing the channel coefficients. In our setting, $Y^0_{b} = A_{k,b}$.
Therefore, the $j$-th bit of the binary representation of $k$ can be estimated as ${k}_j = 0$ if $Y^{j}_{b} / Y^0_{b} = 1$ and ${k}_j = 1$ if $Y^{j}_{b} / Y^0_{b} = -1$.

\DG{Together,} $b_{k}$ (frequency) and 
\LL{$g_{k}^0,\dots,g_{k}^J$ (gains)}
carry the device index.
The correspondence between the device indices and the subcarriers is represented by a bipartite graph with $N$ left nodes and $B$ right nodes. \LL{The $n$-th left node is connected with the $b$-th right node if device $n$ transmits on the $b$-th subcarrier.}

\LL{\begin{definition}\label{def:G}[Bipartite graph induced by active devices in the sparse OFDMA] 
The bipartite graph induced by the active devices in sparse OFDMA consists of $K$ left nodes, corresponding to $K$ active devices, and $B$ right nodes, corresponding to $B$ subcarriers, where the $k$-th left node is connected to the $b$-th right node if device $k$ transmits on subcarrier $b$.
\end{definition}}

We call a subcarrier a \textit{zeroton}, \textit{singleton}, or \textit{multiton}, if no device, a single device, or multiple devices transmit on the subcarrier, respectively. Fig.~\ref{fig:bipartite} illustrates an example of bipartite graph with $B = 5$ subcarriers and $N=4$ devices, $K=3$ of which are active.
 In the example, subcarriers 0, 1, and 4 are zerotons, subcarrier 3 is a singleton, and subcarrier 2 is a multiton.

When there is a single active device in the noiseless setting, the signaling of sparse OFDMA consists of $\lceil \log N \rceil$ OFDM symbols. Therefore, the active device can be identified with code length of $O ((B+M) \log N  )$ samples and computational complexity of $O(B ( \log B ) ( \log N ) )$.

\begin{figure}
  \centering
  \includegraphics[width=6cm]{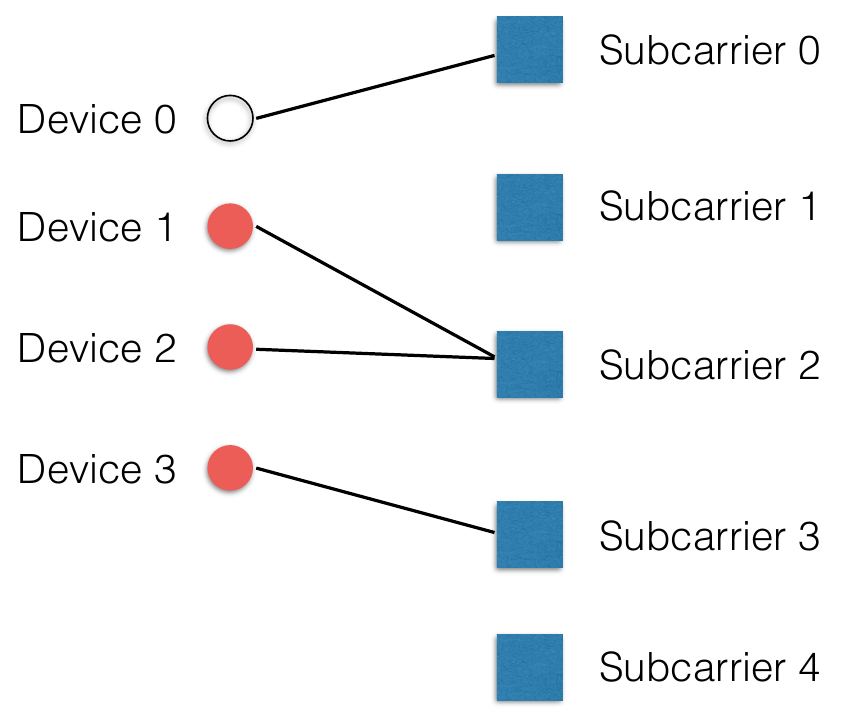}\\
  \caption{Bipartite graph representation of sparse OFDMA. Left nodes represent devices and right nodes represent subcarriers. The active devices are marked in red.  Subcarriers 0, 1 and 4 are zerotons, subcarrier 3 is a singleton, and subcarrier 2 is a multiton.}\label{fig:bipartite}
\end{figure}

\subsection{Identification of Multiple Active Devices With and Without Noise}
\label{sec:identify_multi_device}

When multiple devices are active, the active devices may use colliding subcarriers, so that the device information cannot always be directly recovered from $Y^{j}_{b_{k}} / Y^0_{b_{k}}$. We propose to let the devices transmit  on \textit{multiple} subcarriers. As in the case of a single active device, we first identify active devices from the singletons. The identified device information is then used to bootstrap the detection of other devices.\footnote{A related, simpler setting for device activity detection is group testing, e.g., in \cite{luo2008neighbor}, devices who do not violate the energy levels of all subcarriers are declared to be active.}

\begin{figure}
	\centering
	\includegraphics[width=12cm]{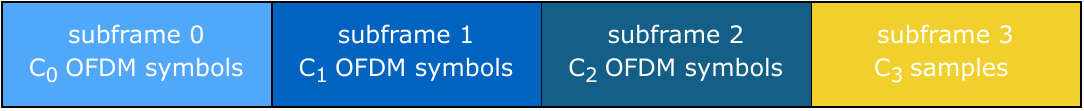}\\
	\caption{\LL{Frame structure of sparse OFDMA. A frame consists of four subframes.}}\label{fig:frame}
\end{figure}

The presence of noise raises additional questions: 1) How can we reliably estimate the channel coefficients? 2) How can we robustly estimate the device information in the noisy setting? 3) How can we \DG{detect whether} 
a subcarrier \DG{is} 
a zeroton, singleton, or multiton? In the following, we further enhance the signaling scheme to address those three challenges.
\LL{Specifically, a frame consisting of four subframes 
is described in Fig.~\ref{fig:frame}, where the 
subframes $0,1$ and $2$ are used for device identification and message decoding, and subframe $3$ is used for delay estimation.} 

\LL{We first introduce the signaling of subframes $0-2$ used for device identification. The subframes consist of $C_0, C_1,$ and $C_2$ OFDM symbols, respectively. Let $C = C_0 + C_1 + C_2$.} In every OFDM symbol, device $k$ is assigned to a fixed set of subcarriers, denoted as $\mathcal{B}_k \subseteq \{ 0, \cdots, B-1 \}$. Specifically, we let device $k$ transmit $\left( \bs^0_{k}, \cdots, \bs^{C}_{k} \right)$, where $\bs^c_{k} = (s^c_{k,0}, \cdots, s^c_{k,B+M-1})$ and
\begin{align}
s^c_{k,i} =   g^c_{k} \sum_{b \in \mathcal{B}_k} \exp \left( \frac{\iota 2 \pi b i}{B} \right),
\end{align}
for $i = 0, \cdots, B+M-1$ and $c = 0, \cdots, C$. 

For the $c$-th OFDM symbol, as in the previous cases, we discard the first $M$ samples and obtain the remaining $B$ samples as $\by^c$. Under the noisy setting, performing $B$-point DFT on $\by^c$ yields
\begin{align}\label{eq:Yk}
Y^c_b &= \sum_{k \in \mathcal{K}: b \in \mathcal{B}_k} A_{k,b} g^c_{k} + W^c_b , \quad b = 0, \cdots, B-1,
\end{align}
where $A_{k,b}$ is given by \LL{\eqref{eq:Akb}},
and $W^c_b$ are i.i.d.\ complex Gaussian variables with distribution $\mathcal{C N} (0, 2 \sigma^2 /B)$. The factor $B$ in the noise variance is due to the integration of $B$ samples in the DFT operation.

Let the design vector for device $k$ be
\begin{align} \label{eq:bgk}
\bg_k = \left(
\begin{array}{c}
\mathbf{1} \\
\tilde{\bg}_k \\
\dot{\bg}_k 
\end{array}
\right)
\end{align}
where the all-one vector $\mathbf{1}$ of length $C_0$, $\tilde{\bg}_{k} \in \mathbb{R}^{C_1}$, and $\dot{\bg}_{k} \in \mathbb{R}^{C_2}$ are the design vectors for the first $C_0$, next $C_1$, and the remaining $C_2$ OFDM symbols, respectively. \LL{The values of $C_0, C_1$, and $C_2$ will be specified in Section~\ref{sec:key_para}. The all-one segment is used for robust estimation of the channel coefficients. The $\tilde{\bg}_k$ segment is used to encode the device index information. We let the entries of $\dot{\bg}_{k}$ be generated according to i.i.d.\ Rademacher ($\pm 1$) variables, with $\mathsf{P} \{\dot{g}_k^c = \pm 1 \} = 1/2$, $c=0, \cdots, C_2-1$.  The \DG{$\dot{\bg}_k$} 
segment is used to mitigate possible false alarms.}

In the absence of noise, we let $\tilde{\bg}_{k} = \left( 1, (-1)^{k_1}, \cdots, (-1)^{k_{\lceil \log N \rceil} } \right)$, which carries the device information. In the noisy setting, the received symbols corresponding to $\tilde{\bg}_{k}$ are corrupted in general. 
In order to robustly estimate the information bits, which is the binary representation of the device's index $k$, we apply a linear error-control code of length \LL{$C_1 =\lceil \lceil \log N \rceil /R   \rceil $} symbols, where $R<1$ is the code rate. Let $\bG \in \mathbb{F}_2^{ \lceil \log N \rceil \times C_1}$ be the generator matrix of the error-control code.
Let $\left( r_{k,0}, \cdots, r_{k, C_1-1}  \right) = \left( {k}_1, \cdots, {k}_{\lceil \log N \rceil } \right) \bG$, where the operation is over the binary field. We construct $C_1 $ OFDM symbols with $\tilde{g}^c_{k} = (-1)^{r_{k,c}}$, $c= 0, \cdots, C_1 - 1$. A good set of codes are the low-complexity capacity approaching codes introduced in~\cite{barg2004error} \DG{(more on this in Section~\ref{sec:pf_step3})}.

\begin{figure}
  \centering
  \includegraphics[width=6cm]{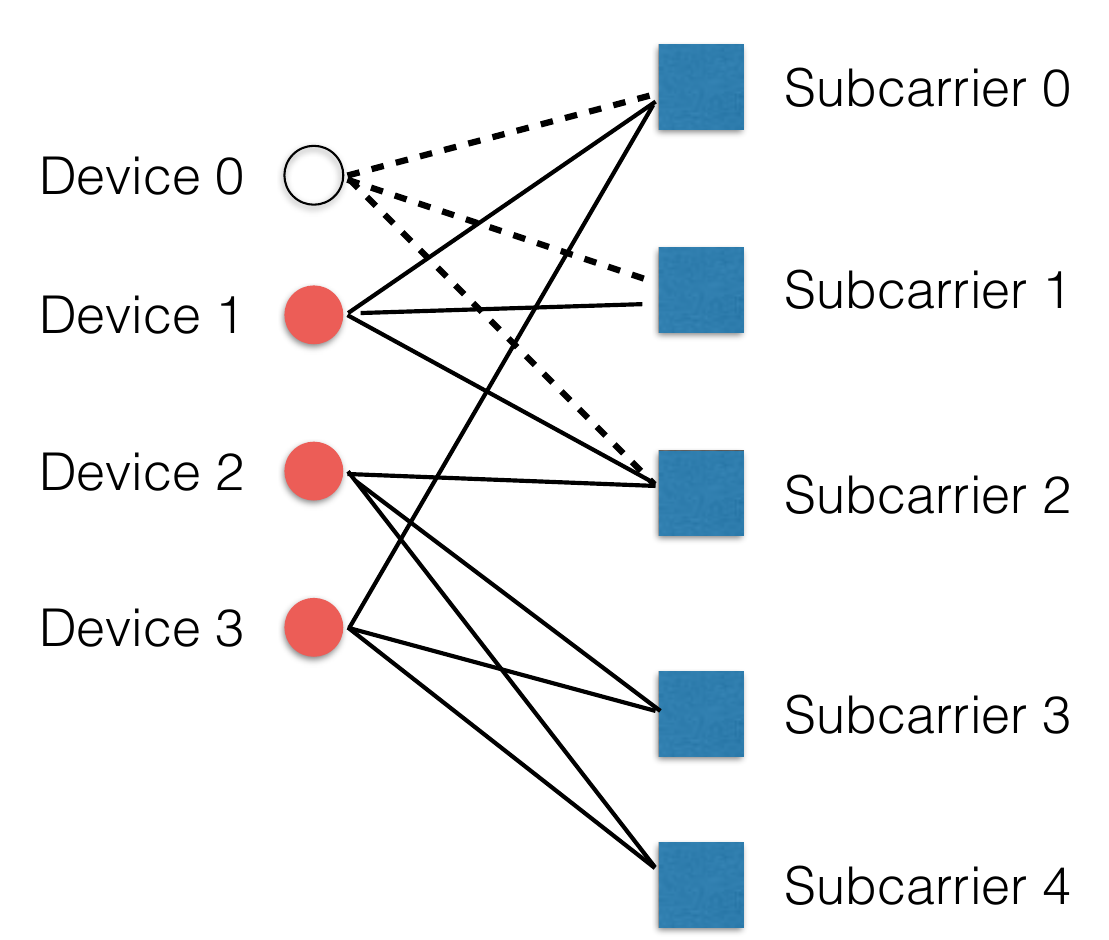}\\
  \caption{Bipartite graph representation of sparse OFDMA with $T=3$. }\label{fig:bipartite2}
\end{figure}

For device $k$, we let $|\mathcal{B}_{k}| = T$ and the set of subcarriers \DG{is equally likely to be any $T$-element subset of}
$\{0 , \cdots, B-1 \}$. \LL{In other words, every device transmits on exactly $T$ out of $B$ subcarriers.} Fig.~\ref{fig:bipartite2} illustrates an example for the case of $T= 3$. 

The last subframe consists of $C_3$ time-domain samples \LL{with low auto-correlation}, which are used to estimate device delays. Delays may be estimated by performing the correlation between the received signal and the pilot samples. In particular, the pilot samples are i.i.d.\ random Rademacher variables with length greater than $M$.

\subsection{Massive Access with Device Identification and Message Decoding}

Similar to embedding the device index information through $\tilde{\bg}_k$, each device can encode both a $\lceil \log N \rceil$-bit device index information and a $\lceil \log S \rceil$-bit message information through $\tilde{\bg}_k$. With a code rate of $R$, we have \LL{$C_1 = \lceil (\lceil \log N \rceil + \lceil \log S \rceil ) /R  \rceil$} OFDM symbols to carry both the device index and the message. The overall design vector of device $k$ is \DG{given by~\eqref{eq:bgk}.}
The DFT values at the $b$-th subcarrier $\bY_b$ is a vector of length $C$ and can be written as
\begin{align} \label{eq:yb}
\bY_b &= \left(
\begin{array}{c}
\bar{\bY}_b \\
\tilde{\bY}_b \\
\dot{\bY}_b 
\end{array}
\right) \\
& = \sum_{k \in \mathcal{K}: b \in \mathcal{B}_k } A_{k,b}  \left(
\begin{array}{c}
\mathbf{1} \\
\tilde{\bg}_{k} \\
\dot{\bg}_{k} 
\end{array}
\right)
+
\left(
\begin{array}{c}
\bar{\bW}_b \\
\tilde{\bW}_b \\
\dot{\bW}_b 
\end{array}
\right)
\end{align}
\DG{where the dimensions of signals $\bar{\bY}_b$, $\tilde{\bY}_b$, and $\dot{\bY}_b$ are $C_0$, $C_1$, and $C_2$, respectively, so are the dimensions of the noise vectors $\bar{\bW}_b$, $\tilde{\bW}_b$, and $\dot{\bW}_b$.}

\section{Device Identification and Message Decoding} 
\label{sec:async}

We first describe a robust subcarrier detection scheme that achieves two goals: 1) It can distinguish whether a subcarrier is a zeroton, a singleton, or a multiton (as defined in Section \ref{sec:device_identify_b}); 2) For singleton subcarriers, it can detect the device index reliably. We then describe the overall identification and decoding scheme.

\subsection{Robust Subcarrier Detection}

We focus on a certain device $k$ that is hashed to subcarrier $b$. The frequency-domain values are given by~\eqref{eq:yb}.

\subsubsection{Channel Phase Estimation}

We first 
estimate the phase of $A_{k,b}$ as
\begin{align}\label{eq:phase}
\LL{\hat{\theta}_b} &= \angle \left( \frac{1}{C_0} \sum\limits_{c=0}^{C_0 - 1} \bar{Y}^c_{ b } \right).
\end{align}
Suppose device $k$ transmits on a singleton subcarrier and $C_0$ is large enough, we can obtain sufficiently accurate estimate of the channel phase.

\subsubsection{Device Identification and Message Decoding}
\label{sec:device_identification}

With the phase estimation \LL{$ \hat{\theta}_b$}, 
we can compensate the phase of $A_{k,b}$ and try to decode the device index information. For each subcarrier $b$, we perform (hard) binary decision on $ \text{Re} \left\{ \tilde{Y}^c_b e^{- \iota \LL{\hat{\theta}_b}} \right\}$ for all $C_1$ symbols, and then decode the index and message. We focus on singleton subcarriers, because this allows us to apply the well-studied point-to-point capacity approaching codes. A more sophisticated multiuser decoding method can also be used. In Section~\ref{sec:proof_sync}, we show that single-user decoding is sufficient to establish the desired scaling laws.

\subsubsection{Singleton Detection and Verification}

First, we declare subcarrier $b$ is a zeroton \LL{if $\Vert\dot{\bY}_b\Vert^2<\eta$, where $\eta$ is some constant threshold}. Otherwise, we attempt to estimate the index using the method described in Section~\ref{sec:device_identification}, assuming subcarrier $b$ is a singleton. Suppose $\hat{k}$ is the estimated index. We perform the following validation process. We estimate the nonzero signal as
\begin{align}\label{eq:chnl_est}
\dot{A}_{\hat{k}, b} = \frac{1}{C_2} \dot{\bg}_{\hat{k}}^{\dagger} \dot{\bY}_b,
\end{align}
where $\dot{g}_k^c $, $c=0, \cdots, C_2-1$, are as described in Section~\ref{sec:identify_multi_device}.
Then we declare that subcarrier $b$ is a singleton if and only if it passes the energy threshold test, i.e.,
\begin{align}\label{eq:multiton_check}
|| \dot{\bY}_b - \dot{A}_{\hat{k},b} \dot{\bg}_{\hat{k}} ||_2^2 \leq  \eta.
\end{align}

The above validation scheme is similar to that used for sparse DFT and sparse Walsh-Hadamard transform (WHT), where the singleton verification approach has been proved to be correct with high probability for signal amplitudes lying in a known discrete alphabet~\cite{pawar2014robust,li2014spright}. In this paper, we further show that it can effectively identify the singletons for \textit{arbitrary} 
\DG{analog} amplitudes that are bounded away from zero, as specified in Theorems~\ref{theorem:noisy_discovery_sync} and~\ref{theorem:noisy_discovery_asyn}.

\subsubsection{Overall Subcarrier Detection}

Based on the preceding process, the robust subcarrier detection algorithm is illustrated in Algorithm~\ref{algo:robust_bin_detect}. For each subcarrier $b$, the DFT values is decomposed into three segments as in~\eqref{eq:yb}. 
If the subcarrier is not a zeroton, we first estimate the phase of channel coefficient based on $\bar{\bY}_b$, which is used to compensate the phase of $\tilde{\bY}_b$. We then estimate the device index as $\hat{k}$. The channel coefficient is estimated as $\dot{A}_{\hat{k},b}$ according to \eqref{eq:chnl_est}. If \eqref{eq:multiton_check} is satisfied, a singleton is declared, the device $\hat{k}$ is estimated to be active, and the message is decoded. Otherwise, a multiton is declared.
 
\begin{algorithm}
\caption{Robust-Subcarrier-Detect ($\bY$)}
\label{algo:robust_bin_detect}
\begin{algorithmic}[]
\STATE \textbf{Input}: Subcarrier values $\bY = \left( \bar{\bY}^{\dagger}, \tilde{\bY}^{\dagger}, \dot{\bY}^{\dagger} \right)^{\dagger}$, where $\bar{\bY} \in \mathbb{C}^{C_0}$, $ \tilde{\bY} \in \mathbb{C}^{C_1}$ and $\dot{\bY}  \in \mathbb{C}^{C_2}$.
\STATE \textbf{Output}: Declaration of zeroton/singleton/multiton, and estimates of index and data if applicable.
\IF {$   || \dot{\bY} ||^2 < \eta $}
\STATE declare zeroton and return.
\ENDIF
\STATE  $  \hat{\theta} \leftarrow  \text{phase} \left(  \mathbf{1}^T \bar{\bY} / C_0 \right).$
\STATE $\tilde{\bZ} \leftarrow  {\rm Re} \{ \tilde{\bY} e^{- \iota \hat{\theta} } \} .$
\STATE $(\hat{k}, \hat{\tilde{g}}_{\hat{k}})  \leftarrow \text{Decode} (\tilde{\bZ}).$
\STATE $\dot{A}_{ \hat{k} } \leftarrow \dot{\bg}_{\hat{k}}^{\dagger} \dot{\bY} / C_2$.
\IF {$|| \dot{\bY} - \dot{A}_{\hat{k}} \dot{\bg}_{\hat{k}} ||_2^2 \leq \eta$} 
\STATE {declare singleton and return $\left( \hat{k}, \hat{\tilde{g}}_{\hat{k}} \right)$.}
\ELSE
\STATE {declare multiton and return.}
\ENDIF
\end{algorithmic}
\end{algorithm}

\subsection{Overall Framework}
\label{sec:overall}

\begin{algorithm}
\caption{Asynchronous Massive Access via Sparse OFDMA}
\label{algo:suuc_cancel}
\begin{algorithmic}[]
\STATE \textbf{Input:} Subcarrier values $\bY_b$, $b=0, \cdots, B-1$.
\STATE \textbf{Output:} Detected active device set $\hat{\mathcal{K}}$ and their messages.
\STATE \textit{Initialize:} Set $\mathcal{B}$ to be the set of all subcarriers. Set $\hat{\mathcal{K}} $ and $\mathcal{L}$ to be the empty set. 
\FOR {every subcarrier $b$}
\IF {Robust-Subcarrier-Detect $\left( {\bY}_b\right)$ declares a singleton $\left( \hat{k}, \hat{\tilde{\bg}}_{\hat{k}} \right)$ }
\STATE Add $\hat{k}$ to $\mathcal{L}$.
\ENDIF
\ENDFOR
\WHILE {fewer than $K$ iterations and $\mathcal{L}$ is not empty}
\STATE Pick arbitrary $k \in \mathcal{L}$, remove $k$ from $\mathcal{L}$, and add $k$ to $\hat{\mathcal{K} } $ .
\STATE Estimate $\hat{m}_{k}$ and $\hat{a}_{k}$ according to~\eqref{eq:est_delay} and \eqref{eq:channel_estimate}.
\STATE Set $\mathcal{S}$ to be the set of subcarriers in $\mathcal{B}$ that are connected with $k$.
\FOR {every subcarrier $b' \in \mathcal{S}$}
\STATE Cancel the signal of device $k$ from subcarrier $b'$ using \eqref{eq:succ_canel}.
\IF {Robust-Subcarrier-Detect  $\left({\bY}_{b'} \right)$ declares a zeroton or a singleton $ \left( \hat{k}, \hat{\tilde{g}}_{\hat{k}} \right)$}
\STATE Remove $b'$ from $\mathcal{B}$.
\STATE Add $\hat{k}$ to $\mathcal{L}$ (if a singleton is declared).
\ENDIF
\ENDFOR
\ENDWHILE
\end{algorithmic}
\end{algorithm}

Once a device index is estimated based on a singleton, its contributions to all its connected subcarriers \DG{are} 
canceled out, which may result in new singleton subcarriers. For example, in Fig.~\ref{fig:bipartite2}, device 1 is first detected from the singleton subcarrier 0 and its values are subtracted from subcarrier 2. Then subcarrier 2 becomes a singleton subcarrier and device 2 can be detected from it. There is, however, one challenge. The DFT values from each device at a subcarrier depends on its delay due to \eqref{eq:Akb}. We need to estimate the delay using subframe 3 of the received signal in order to perform successive cancellation.

Let $\mathcal{I}$ denote the time-domain indices corresponding to the last subframe with the first $M$ samples skipped. Define the decision statistic based on $\mathcal{I}$ as:
\begin{align}\label{eq:metric}
\mathcal{T}_k (m) =  \sum_{i \in \mathcal{I}}  y_{i+m} s_{k,i+M}^{\ast} .
\end{align}
The reason of discarding the first $M$ samples in the \LL{last} subframe is to guarantee that the cross-correlation of the pilot sequences between different users is performed entirely on the Rademacher random variables.
We estimate the delay of device $k$ as
\begin{align}\label{eq:est_delay}
\hat{m}_{k}  \in \arg\max_{m = 0, \cdots, M-1} \text{Re} \{ \mathcal{T}_k (m) \}. 
\end{align}
The channel coefficient is then estimated to be 
\begin{align}\label{eq:channel_estimate}
\hat{a}_{k} = \frac{1}{C} \bg_k^{\dagger} \bY_b e^{-\frac{\iota 2 \pi b (M-\hat{m}_k)}{B} }.
\end{align}
The DFT values of a connected unprocessed subcarrier $b'$ are then updated according to
\begin{align}\label{eq:succ_canel}
\bY_{b'} \leftarrow \bY_{b'} -  \hat{a}_{k} \exp \left( \frac{\iota 2 \pi b' (M - \hat{m}_{k} ) }{B} \right) \bg_{k}.
\end{align}

The overall massive access framework is described in Algorithm~\ref{algo:suuc_cancel}. Throughout the algorithm, we maintain three lists: $\hat{ \mathcal{K} }$ is a list storing the estimated device indices, $\mathcal{L}$ is a list of detected devices for cancellation, and $\mathcal{B}$ is a list of surviving subcarriers.
%
We first detect all singleton subcarriers and identify their corresponding devices using Algorithm~\ref{algo:robust_bin_detect}. Then we successively cancel each identified device, potentially exposing more singletons along the way to allow more devices to be identified. This process continues until no singleton subcarrier is left. 

We next prove that the preceding sparse OFDMA can efficiently identify the active devices and decode their messages.  We will first prove the synchronous case ($M=0$) in Section~\ref{sec:proof_sync}, and then extend the proof to the asynchronous case in Section~\ref{sec:proof_asyn}. 

\section{Proof of Theorem \ref{theorem:noisy_discovery_sync} (the Synchronous Case)}
\label{sec:proof_sync}

\subsection{Key Parameters and Propositions}\label{sec:key_para}


For codeword construction, we choose constants
\begin{equation}\label{eq:para_start}
T\geq 3 
\end{equation}
and
\begin{equation}\label{eq:para_beta}
\LL{\beta_0\geq T(T-1)+1,}
\end{equation}
e.g., by letting $T=3$ and \LL{$\beta_0 = 7$}. 
We also set the following parameters as 
\begin{align}
\label{eq:para_B} B & = \beta_0 K \\
C_0 &= \lceil \log N \rceil \label{eq:C0} \\
 C_1 &= \LL{\lceil  (\lceil \log N \rceil + \lceil \log S \rceil ) / R \rceil }\\
\label{eq:para_end}  C_2 &=\Lina{\lceil\beta_1 \log K \rceil},
\end{align}
\Lina{where $\beta_1 > 0$ is a constant to be} specified in Appendix~\ref{append:proof_singleton0},
and $R<1$ is \LL{a constant denoting the rate of a low-complexity capacity-approaching code for the binary-symmetric channel (BSC) to be further explained in Appendix~\ref{append:proof_singleton1}.} 
The total number of OFDM symbols is
\begin{align}
\Lina{C \leq (1 + \lceil 1 / R \rceil)  \lceil \log N \rceil +  \lceil 1 / R \rceil \lceil \log S \rceil + \lceil \beta_1\log K \rceil.}
\end{align}
\DG{The codelength $L=BC=\beta_0KC$ thus satisfies~\eqref{eq:Lsync} as long as $\beta_0$, $\beta_1$, and $R$ are positive constants.}


\LL{
\DG{In this proof, we shall invoke several results about hypergraphs.}
\XC{A hypergraph is a generalization of a graph in which an edge can join any number of vertices. An edge in a hypergraph is also called a hyperedge.
A hypergraph is a hypertree if the graph is a tree, i.e., no cycle exists. A hypergraph is a unicyclic component if the graph contains only one cycle. A $T$-uniform hypergraph is a hypergraph where all the hyperedges have degree-$T$.
}\DG{Here we let} the device nodes \DG{become} 
hyperedges and the subcarriers \DG{become hypergraph} 
vertices.  
\DG{A} hyperedge is incident on a vertex if the corresponding device node is connected to the corresponding subcarrier. 
}

\LL{Let $G$ denote the bipartite graph induced by the active devices (see Defition~\ref{def:G}). 
We now characterize the probability that $G$ is such that no two active devices share the same set of subcarriers, so that it is convertible to an equivalent hypergraph (denoted as $G\in\hat{\mathcal{G}}$). Since there are $\binom{B}{T}$ distinct subsets of $T$ subcarriers, we have
\begin{align}
  \mathsf{P}\{G\in\hat{\mathcal{G}}\}
  =
  \binom{B}{T}\left[\binom{B}{T}-1\right]\cdots\left[\binom{B}{T}-K+1\right]
  \Bigg/
  \left(\binom{B}{T}\right)^K.
\end{align}
Since $\binom{B}{T}\geq O(K^3)$ for $T\geq 3$, we have
\begin{equation}
   \mathsf{P}\{G\in\hat{\mathcal{G}}\}\geq 1-\frac{\zeta}{K}
\end{equation}
with some constant $\zeta$ \cite[The birthday problem]{feller1957introduction}. }


 By the construction of the sparse OFDMA in Section \ref{sec:identify_multi_device}, every device is associated with exactly $T$ subcarriers, so the hypergraph induced by the active devices is a (random) $T$-uniform hypergraph.
\begin{definition}\label{def:gensemble}[An ensemble of hypergraphs]
Let $\mathcal{G}$ denote the ensemble of $T$-uniform hypergraphs with $K$ hyperedges and $B$ vertices that consist of components each of which is either a hypertree or a unicyclic component, where no component has more than \LL{$\alpha T \log (\beta_0K)/ (T-1)$ hyperedges} with some constant $\alpha$.
\end{definition}

In the following, we will first show that there exists \LL{$\alpha$ such that $G$ is convertible to a hypergraph in $\mathcal{G}$ (denoted as $G \in \mathcal{G}$)} with high probability (Proposition~\ref{claim:step1}).
We then characterize the error propagation effects in the case of $G \in \mathcal{G}$ (Proposition~\ref{claim:step2}) and show that the Robust-Subcarrier-Detect makes the correct decision with high probability (Proposition~\ref{claim:step3}). It follows then that synchronous massive access via Algorithm~\ref{algo:suuc_cancel} succeeds with high probability (Proposition~\ref{claim:step4}). \XC{The following propositions will be proved in Sections~\ref{sec:proof_claim1} to
~\ref{sec:pf_step3}.}


\begin{proposition}\label{claim:step1}
Under \eqref{eq:para_start}-\eqref{eq:para_end}, \LL{there exist constants $K_0'>0$, and \LL{$\nu >1$} dependent on $\beta_0$, such that for every $K \geq K_0'$},
the hypergraph $G$ satisfies $\mathsf{P} \left\{ G \in \mathcal{G} \right\} \geq 1 - \LL{\nu} / K$.
\end{proposition}

\begin{proposition}\label{claim:step4}
If $G\in\mathcal{G}$, then Algorithm~\ref{algo:suuc_cancel} will detect all active devices and decode their messages correctly as long as during its execution Algorithm~\ref{algo:robust_bin_detect} always makes the correct decision.
\end{proposition}

\begin{proposition}\label{claim:step2}
\LL{Suppose \eqref{eq:para_start}-\eqref{eq:para_end} hold,} \XC{$G \in \mathcal{G}$, and Algorithm~\ref{algo:robust_bin_detect}  always makes the correct decisions in the previous $t-1$ iterations during the execution of Algorithm~\ref{algo:suuc_cancel}. Let $S(t-1)$ be the set of recovered devices. At iteration $t$, the DFT value at subcarrier $b$ \DG{can be expressed as} 
\begin{align}\label{eq:single_Y}
\bY_{b} = \sum_{k \in \mathcal{K} \backslash S(t-1): b \in \mathcal{B}_k} A_{k,b} \bg_k + \bW_{b} + \bV_{b},
\end{align}
where the sum is over the set of active devices that are hashed to subcarrier $b$ and not yet recovered, and $\bV_b$ is due to the residual channel estimation errors from the recovered devices.  Moreover, conditioned on the design parameters $\bg_{\ell}$ of the previously identified devices, each entry of $\bV_{b}$ is Gaussian variable with zero mean and variance bounded by \LL{$\frac{\beta_2\sigma^2}{B}$ with a constant $\beta_2$.}}
\end{proposition}

\begin{proposition}\label{claim:step3}
\LL{Suppose \eqref{eq:para_start}-\eqref{eq:para_end} hold and $G \in \mathcal{G}$. There exists a constant $\eta$ such that for every $t=1,2,\cdots$, conditioned on that Algorithm~\ref{algo:robust_bin_detect} makes correct decisions in the first $t-1$ iterations during the execution of Algorithm~\ref{algo:suuc_cancel}, Algorithm~\ref{algo:robust_bin_detect} makes a wrong decision with probability no greater than $7K^{-2}$ in the $t$-th iteration.}
\end{proposition}

\LL{Assuming Propositions~\ref{claim:step1}-\ref{claim:step3} hold, we upper bound the massive access error probability $P_s$ as follows.} Let $\mathcal{E}$ denote the event that Algorithm~\ref{algo:robust_bin_detect} makes at least one wrong decision during the execution of Algorithm~\ref{algo:suuc_cancel}. By Proposition~\ref{claim:step4}, massive access succeeds if $G \in \mathcal{G}$ and that $\mathcal{E}$ does not occur. Evidently,
\begin{align}
\LL{\label{eq:supp_fail}P_s\leq \mathsf{P} \{ \mathcal{E} | G \in \mathcal{G}\} + \mathsf{P} \{ G \notin \mathcal{G}|G\in\hat{\mathcal{G}}  \}+\mathsf{P}\{G\notin\hat{\mathcal{G}}\}.}
\end{align}

Every time a device is recovered, Algorithm~\ref{algo:robust_bin_detect} is performed on its connected subcarriers. Since there are $K$ active devices and each of them is connected to $T = O(1)$ subcarriers, Algorithm~\ref{algo:robust_bin_detect} runs for at most $K T$ times throughout the detection process. By the union bound and the result of Proposition~\ref{claim:step3},
\begin{align}
\mathsf{P} \{\mathcal{E} | G \in \mathcal{G} \} &\leq  1- \left( 1-\frac{7}{K^2} \right)^{KT}  \\
& \leq \frac{7 T}{K}\label{eq:union_bin_error}.
\end{align}


Combining Proposition~\ref{claim:step1}, \eqref{eq:supp_fail}, and \eqref{eq:union_bin_error}, massive access fails with probability 
\LL{\begin{align}
P_s \leq (7 T + \nu+\zeta) / K.
\end{align}
Therefore, given the choice of $T$, $B$, and $C$, massive access fails with probability less than $\epsilon$ for $K \geq \max\{(7 T + \nu+\zeta)/\epsilon,K_0'\}$. 
}

\subsection{Proof of Proposition~\ref{claim:step1}}
\label{sec:proof_claim1}


\LL{Let $\mathcal{G}_0$ denote the ensemble of $T$-uniform hypergraphs with $K$ hyperedges and $B$ vertices consisting of only hypertrees and unicylic components. For a given constant $\alpha$, let $\mathcal{G}_1$ denote the ensemble of hypergraphs with the largest component containing fewer than $\alpha T\log (\beta_0K) / (T-1)$ hyperedges, and $\mathcal{G}_{2}$ denote the ensemble of hypergraphs with the largest component containing fewer than $\alpha T \log B$ \XC{vertices}. Evidently, $\mathcal{G}=\mathcal{G}_0\cap\mathcal{G}_1$.}

\LL{Due to our choice of $K$ and $B$, the $T$-uniform hypergraph is in the so-called subcritical phase, which guarantees some simple structural properties. In particular, \cite[Theorem 4]{karonski2002phase} establishes that the hypergraph is entirely composed of hypertrees and unicyclic components with high probability. To be precise, the proof in \cite{karonski2002phase} \DG{implies} that}
\LL{\begin{equation}\label{eq:g0}
\mathsf{P}\{G\in\mathcal{G}_0\}\geq 1-\frac{\nu-1}{K}
\end{equation}
\DG{for} some constant \DG{$\nu>1$} depend\DG{ent} 
on $\beta_0$.}

\LL{We next show that there exists a constant $\alpha$, such that $G\in\mathcal{G}_2$ with high probability.} The size of a hypergraph is defined as the number of hyperedges $K$ in the graph. The number of vertices in the hypergraph is $B$. \XC{The average vertex degree of a vertex $v$ is defined as the expected number of pairs of $(v_i, e_i)$, where $v_i$ and $e_i$ are some vertex and hyperedge in the graph, such that $(v,v_i)$ are connected via hyperedge $e_i$.}

\begin{lemma} \label{lemma:hypergraph}
\LL{$G$} has an average vertex degree of $K T (T-1) /B$.
\end{lemma}

Proof of Lemma~\ref{lemma:hypergraph}: 
Consider a subcarrier $b$ (a vertex). Let $X_k =1$ if device $k$ uses subcarrier $b$, and $X_k=0$, otherwise. \LL{Device $k$ has $\binom{B}{T}$ equally likely choices for its subcarriers, where $\binom{B-1}{T-1}$ of those choices include subcarrier $b$.} The average vertex degree is thus calculated as 
\begin{align}
\mathsf{E} \left\{  \sum_{k \in \mathcal{K}} (T-1) X_k \right\} 
&=  K (T-1) \mathsf{E} \left\{ X_k \right\} \\
&=K(T-1) \frac{\binom{B-1}{T-1}}{\binom{B}{T}} \\
& = \frac{K T (T-1)}{ B}.
\end{align} \qed


\XC{
Consider a random hypergraph with $B$ vertices and range $T$.{\footnote{The size of the largest edge is called range of the hypergraph. A $T$-uniform hypergraph has a range of $T$.}} Since $B \LL{>} K T (T-1)$ due to \eqref{eq:para_beta} and \eqref{eq:para_B}, the average {\em vertex degree} is less than 1 by Lemma~\ref{lemma:hypergraph}.
\LL{Then it has been shown in \cite[Theorem 3.6]{schmidt1985component} that, since $G$'s average vertex degree is less than 1, \LL{for $K\geq K_0'$ with a large enough $K_0'$,} there exists a constant $\alpha$, such that}
\begin{equation}\label{eq:g2}
\LL{\mathsf{P}\{G\in\mathcal{G}_2\}\geq1-\frac{1}{K}.}
\end{equation}

Let $r$ and $s$ be the number of vertices and hyperedges of a connected component, respectively. If the component is a hypertree, $r = (T-1)s +1$. If the component is a unicylic component, $r= (T-1)s$. Therefore, when $G \in  \mathcal{G}_0$ and $ G \in \mathcal{G}_2$, the number of hyperedges in the largest component is upper bounded as 
\begin{align}
s &\leq r/(T-1) \\
& \leq \alpha T \log B/(T-1)\\
&=\alpha T\log(\beta_0 K)/(T-1).
\end{align}
\LL{It implies 
\begin{equation}\label{eq:gcap}
\mathcal{G}_0\cap\mathcal{G}_2\subset\mathcal{G}_1.
\end{equation}} Therefore, we have
\LL{
\begin{align}
\mathsf{P} \left\{ G \in \mathcal{G} \right\} &= \mathsf{P} \left\{  G \in  (\mathcal{G}_0 \cap  \mathcal{G}_1) \right\}\label{eq:g0g1} \\
&\geq \mathsf{P} \left\{G\in(\mathcal{G}_0\cap\mathcal{G}_2) \right\}\label{eq:g0g2}\\
&\geq 1-\frac{\nu}{K},\label{eq:Gbound}
\end{align}
where \eqref{eq:g0g1} follows from Definition \ref{def:gensemble} and definitions of $\mathcal{G}_0$ and $\mathcal{G}_1$, \eqref{eq:g0g2} is due to \eqref{eq:gcap}, and \eqref{eq:Gbound} is due to \eqref{eq:g0} and \eqref{eq:g2}.}}

\subsection{Proof of Proposition~\ref{claim:step4}}
\label{sec:pf_step4}

Every graph in $\mathcal{G}$ consists of only hypertrees and unicyclic components. By \cite[Lemma 3.4]{price2011efficient}, every hypertree has a hyperedge (device node) incident on at least $T-1$ singleton subcarriers. Every unicyclic component has a hyperedge that is incident on either $T-1$ or $T-2$ singleton subcarriers. \XC{ When $T \geq 3$, every hypergraph in $\mathcal{G}$ always has singleton subcarriers.
Assuming Algorithm~\ref{algo:robust_bin_detect} correctly detects the singleton subcarriers and estimates the devices' indices,} the detected devices will be removed from the hypergraph. Since removing any device node or hyperedge in $G \in \mathcal{G}$ still yields a hypergraph in $\mathcal{G}$,  Algorithm~\ref{algo:suuc_cancel} will continue until all active devices are correctly detected.

\subsection{Proof of Proposition~\ref{claim:step2}}

We make use of the error propagation graph proposed in \cite{chen2016generalized} to characterize the residual channel estimation errors.  \XC{The error propagation graph is a directed subgraph induced by the device identification Algorithm~\ref{algo:suuc_cancel}. It begins from some singleton subcarriers. Every singleton subcarrier points to the device node that should be detected based on it. Every device node in the graph points to all subcarriers it is then cancelled from. Fig.~\ref{fig:err_prop} illustrates the error propagation subgraph concerning device 2. In the error propagation graph, device 0 is estimated from singleton subcarrier-a and its values are cancelled from the connected subcarrier-b and subcarrier-c. In a subsequent iteration, subcarrier-b becomes a singleton. Device 1 can be detected and its values are cancelled from subcarrier-c. In a third iteration, subcarrier-c becomes a singleton and device 2 is detected. 
  }

\begin{figure}
  \centering
  \includegraphics[width=7cm]{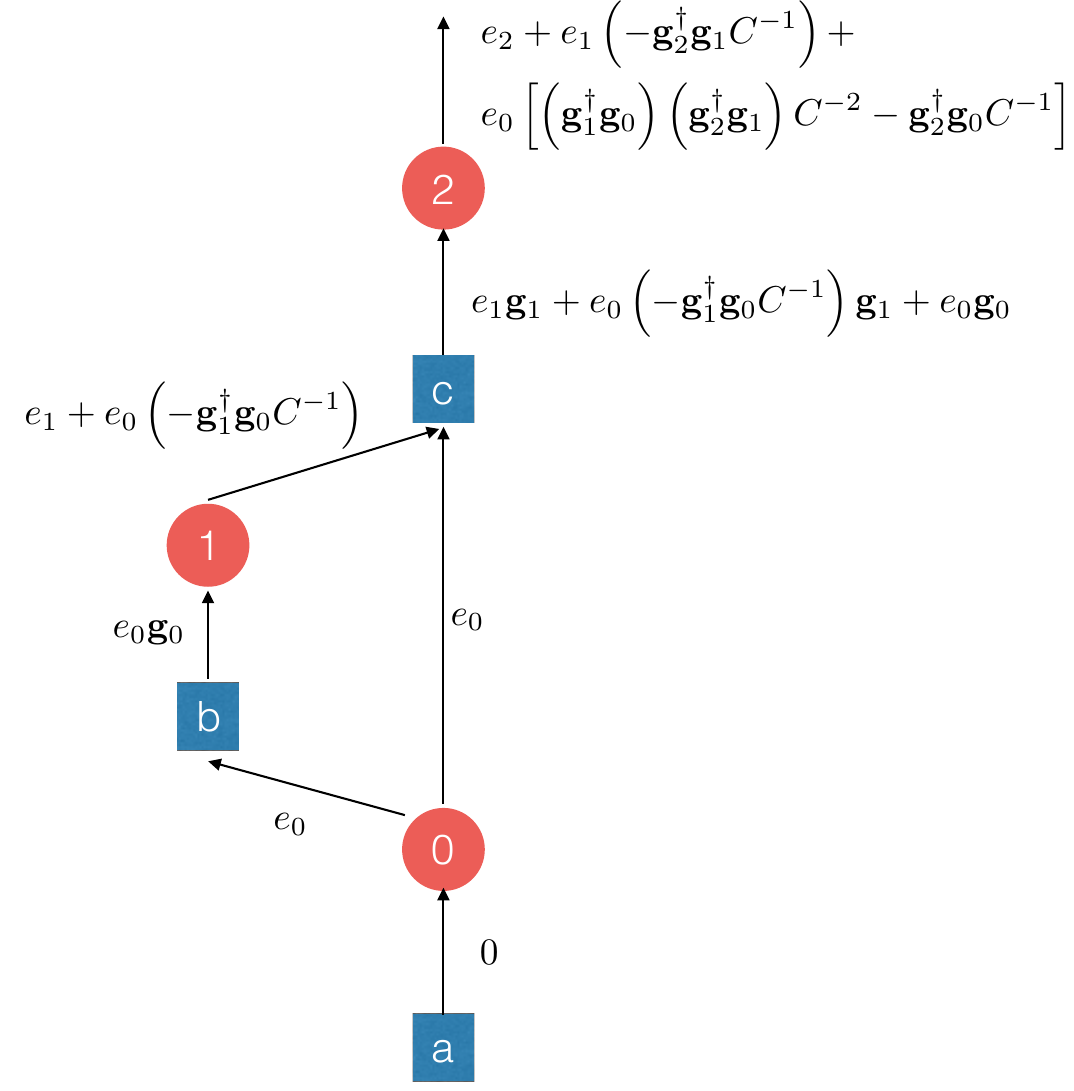}\\
  \caption{Error propagation graph for device 2. Device 0 is detected from subcarrier-a at iteration $t=1$, device 1 is detected from subcarrier-b at iteration $t=2$, and device 2 is detected from subcarrier-c at iteration $t=3$. }\label{fig:err_prop}
\end{figure}

Let $b_k$ be the immediate singleton subcarrier used to recover the device $k$.
Let $A_k$ be the channel estimation error of device $k$ defined as 
\begin{align}\label{eq:residual_error_k}
A_k &= \hat{A}_{k,b_k}- A_{k,b_k} ,
\end{align}
where $b_k$ is the subcarrier which will be used to estimate the device index $k$, $A_{k,b_k}$ is given by~\eqref{eq:Akb} and 
\begin{align}
\hat{A}_{k,b_k} = \bg^{\dagger}_k \bY_{b_k}/ C
\end{align} 
is the estimate according to~\eqref{eq:channel_estimate} in this case with no delay. 

We will keep track of the errors using the graph. The error $A_k$ causes cancellation error $A_k g_k$ to every downstream subcarrier node. The errors accumulated by a singleton subcarrier adds to the estimation errors of its downstream device nodes. 

Consider the first iteration and a singleton subcarrier $b_k$ due to device $k$. The DFT value of subcarrier-$b_k$ is given by
\begin{align}
\bY_{b_k} &= A_{k,b_k} \bg_k + \bW_{b_k}, 
\end{align}
so the residual estimation error is given by 
\begin{align}
A_k & = C^{-1} {\bg}_k^{\dagger} {\bW}_{b_k}.
\end{align}
%

At iteration $t \geq 2$, with successive cancellation, the updated frequency value at subcarrier $b$ is
\begin{align}
\bY_{b} &= \sum_{k \in \mathcal{K} \backslash S(t-1): b \in \mathcal{B}_k} A_{k,b} \bg_k +  \bW_{b} + \sum_{\ell \in S(t-1): b \in \mathcal{B}_{\ell}} \left(  A_{\ell, b_{\ell}} - \hat{A}_{\ell, b_{\ell}} \right) \bg_{\ell} \\
&= \sum_{k \in \mathcal{K} \backslash S(t-1): b \in \mathcal{B}_k} A_{k,b} \bg_k +  \bW_{b} - \sum_{\ell \in S(t-1): b \in \mathcal{B}_{\ell}} A_{\ell} \bg_{\ell}.
\end{align}
Suppose some subcarrier $b_k$ is a singleton due to device $k$ and is used to recover the index, then the channel estimation error is calculated as
\begin{align}
A_k &=  C^{-1} \bg_k^{\dagger} \bY_{b_k} - A_{k,b_k}  \\
& = C^{-1} \bg_k^{\dagger} \left( A_{k,b_k} \bg_k + \bW_{b_k} - \sum_{\ell \in S(t-1): b_k \in \mathcal{B}_{\ell}} A_{\ell} \bg_{\ell} \right) - A_{k,b_k} \\
\label{eq:pk_last}& = C^{-1} {\bg}_k^{\dagger} {\bW}_{b_k} -  \sum_{\ell \in S(t-1): b_k \in \mathcal{B}_{\ell}}  C^{-1} A_{\ell} \bg_k^{\dagger} \bg_{\ell}. 
\end{align}

Using recursion, the estimation error of $A_{k,b}$ for $k \in S(t) \backslash S(t-1)$ is calculated as
\LL{\begin{align}
A_k=e_k+\sum_{(\ell_0,\cdots,\ell_i)\in \mathcal{P}(k),\ell_0\in S(t-1)}e_{\ell_0}(-\bg^{\dagger}_k\bg_{\ell_{i}}/C)(-\bg^{\dagger}_{\ell_{i}}\bg_{\ell_{i-1}}/C)\cdots(-\bg^{\dagger}_{\ell_1}\bg_{\ell_0}/C),
\end{align}
where $\mathcal{P}(k)=\{(\ell_0,\cdots,\ell_i):\ell_0,\cdots,\ell_i,k\text{ is a path of devices in the error propagation graph}\}$, and $e_{\ell_0}$ is given by
\begin{align}
e_{\ell_0} & = C^{-1} {\bg}_{\ell_0}^{\dagger} {\bW}_{b_{\ell_0}}.
\end{align}
It is easy to see that every $e_{\ell_0}$ is independent of the design parameters $\{\bg_k\}$ of all the devices and is distributed according to $\mathcal{CN} (0, 2 \sigma^2 / (B C) )$.}

In Fig.~\ref{fig:err_prop}, for instance, when calculating the estimation error $A_2$ for device node $2$, we take into account device nodes $0$ and $1$ that have been recovered. \LL{In particular, $\mathcal{P}(2)=\{(0,2),(1,2),(0,1,2)\}$. For path $(0,2)$, the product term in the summation is $-e_0\bg_2^{\dagger}\bg_0/C$. For path $(1,2)$, the product term in the summation is $-e_1\bg_2^{\dagger}\bg_1/C$. For path $(0,1,2)$, the product term in the summation is $e_0\bg_2^{\dagger}\bg_1\bg_1^{\dagger}\bg_0/C^2$.} Consequently, we have
$A_2=e_2+e_0(-\bg_2^{\dagger}\bg_0/C+\bg_2^{\dagger}\bg_1\bg_1^{\dagger}\bg_0/C^2)-e_1\bg_2^{\dagger}\bg_1/C$.

Therefore, the DFT value at subcarrier $b$ is calculated as 
\begin{align}
\bY_{b} = \sum_{k \in \mathcal{K} \backslash S(t-1): b \in \mathcal{B}_k} A_{k,b} \bg_k +  \bW_{b} + \bV_{b},
\end{align}
where 
\LL{\begin{equation}\label{eq:Vb}
\mathbf{V}_b=-\sum_{(\ell_0,\cdots,\ell_i)\in \mathcal{P}'(b),\ell_0\in S(t-1)}e_{\ell_0}(-\bg^{\dagger}_{\ell_i}\bg_{\ell_{i-1}}/C)\cdots(-\bg^{\dagger}_{\ell_1}\bg_{\ell_0}/C)\bg_{\ell_i}
\end{equation}
where $\mathcal{P}'(b)=\{(\ell_0,\cdots,\ell_i):\ell_0,\cdots,\ell_i\text{ is a path of devices to subcarrier $b$, i.e., device $\ell_i$ points to subcarrier $b$, }\\
\text{in the error propagation graph}\}$.}

\LL{Suppose $G\in\mathcal{G}$, then $|\mathcal{P}'(b)| \leq 2$.} Moreover, by Proposition~\ref{claim:step1}, the number of left nodes in each component is less than \LL{$\alpha T\log (\beta_0K) / (T-1)$}, which indicates that \LL{$|S(t-1)|\leq \alpha T\log (\beta_0K) / (T-1)\leq \alpha T\log K/(T-1)+2\alpha\log\beta_0$}. Since each entry of the design parameter $\bg_{\ell}$ is i.i.d. Rademacher variable and $\bg_{\ell}\in\mathbb{R}^C$, $|-\bg^{\dagger}_{\ell_i}\bg_{\ell_{i-1}}/C|\leq 1$. 
\LL{Thus, for each entry of $\bV_b$, $e_{\ell_0}$ has the coefficient satisfying $\vert(-\bg^{\dagger}_{\ell_i}\bg_{\ell_{i-1}}/C)\cdots(-\bg^{\dagger}_{\ell_1}\bg_{\ell_0}/C)g_{\ell_i}^c\vert\leq 1$. Combining the fact that $e_{\ell_0}\sim\mathcal{CN}(0,2\sigma^2/(BC))$, each entry of $\bV_b$ is a Gaussian variable with zero mean and variance bounded by $4|S(t-1)|\cdot\frac{2\sigma^2}{BC}$. Since \begin{align}
\frac{8|S(t-1)|}{C}\leq&\frac{8\alpha T\log K/(T-1)+16\alpha\log\beta_0}{\log N+\log N/R+\beta_1\log K}\\
\leq&\frac{8\alpha T/(T-1)+16\alpha\log\beta_0}{1+1/R+\beta_1}\label{eq:beta2}
\end{align}
according to \eqref{eq:C0}-\eqref{eq:para_end}, the variance is bounded by $\frac{\beta_2\sigma^2}{B}$, where $\beta_2$ equals to the right hand side of \eqref{eq:beta2} and is a constant.}



\subsection{Proof of Proposition~\ref{claim:step3}}
\label{sec:pf_step3}


As described in the proof of Proposition~\ref{claim:step2}, the frequency domain signal in subcarrier-$b$ can be written as \eqref{eq:single_Y}.
The detection error depends on $\bV_b$ and hence on the number of devices that cause interference. 
Let 
\begin{align}
\bZ_b = \bW_b + \bV_b
\end{align}
denote the interference plus noise. We set the energy thresholds $\eta$.

\subsubsection{Zeroton Error Detection} Suppose subcarrier $b$ is a zeroton. The DFT values in $\bY_b$ are composed of purely noise and interference, i.e., $\bY_b = \bZ_b$. The zeroton error $E_{0}$ occurs only if $|| \dot{\bY}_b||$ is greater than the threshold $\eta$. 
Let $\bgg$ denote the set of design parameters $\bg_k$ of all the previously identified devices. Conditioned on $\bgg$, each entry of $\bZ_b$ is distributed according to $\mathcal{CN}(0, 2 \sigma_z^2)$, where $\sigma_z \leq (1 + \LL{\beta_2}/2) \sigma^2 / B$. The probability of error satisfies
\begin{align}
\mathsf{P} \left\{ E_{0} \right\} &= \mathsf{E} \left\{  
\mathsf{P} \left\{ || \dot{\bZ}_{b} ||^2 \geq \eta \big| \bgg \right\}  \right\}\label{eq:prob_cond_g_step1} \\
\label{eq:prob_cond_g_step2} & \leq  \mathsf{E} \left\{  2 C_2   \mathsf{P} \left\{ | \text{Re} \{\dot{Z}_{b,c}\} |^2 \geq \frac{\eta}{2 C_2} \big| \bgg  \right\}  \right\}\\
& \leq 4 C_2 e^{- \frac{\eta}{4 \sigma_z^2 C_2} } \\
\label{eq:prob_cond_g}& \leq 4 C_2 e^{- \frac{B \eta}{(4 + 2 \LL{\beta_2}) \sigma^2 C_2} },
\end{align}
where \eqref{eq:prob_cond_g_step2} is due to that the real and imaginary elements of $\dot{\bZ}_b$ are i.i.d. Gaussian variables, and moreover if each element has an absolute value bounded by $\eta / (2 C_2)$, then the norm of $\dot{\bZ}_b$ must be bounded by $\eta$.

With $B$ and $C_2$ chosen according to \eqref{eq:para_B} and \eqref{eq:para_end}, respectively, if we pick \LL{$\eta$ as a constant that satisfies
\begin{equation}\label{eq:eta}
\eta\geq\frac{(4+2\LL{\beta_2})\sigma^2\beta_1\lceil\log K\rceil\ln(4K^2\beta_1\lceil  \log K\rceil)}{\beta_0 K},
\end{equation}}then the error probability satisfies
\begin{align}
\label{eq:prob_e0} \mathsf{P} \{ E_{0} \}   &\leq  K^{-2}.
\end{align}
\LL{Note that the right hand side of \eqref{eq:eta} vanishes as $K\rightarrow\infty$. Hence, there exists such a constant $\eta$.}

\subsubsection{Singleton Error Detection} Suppose subcarrier $b$ is a singleton due to device $k$. Let $E_{1}$ denote the subcarrier detection error. A singleton detection error occurs due to one or more of three events: (1) $E_{1,0} = \left\{ || \dot{\bY}_b ||_2^2 < \eta \right\}$; (2) Not $E_{1,0}$ and $E_{1,1} = \left\{  \hat{\tilde{\bg}}_k \neq \tilde{\bg}_k  \right\}$; (3) $E_{1,2} = \left\{ || \dot{\bY}_b - \dot{A}_{k,b} \dot{\bg}_{k}||_2^2 > \eta \right\}$. Thus $E_{1} \subseteq E_{1,0} \cup E_{1,1} \cup E_{1,2}$.


\LL{We prove
\begin{align}
&\mathsf{P} \left\{ E_{1,0} \right\} \leq 2 K^{-2},\label{eq:prob_e10}\\
&\mathsf{P}\left\{ E_{1,1} \right\} \leq K^{-2},\label{eq:prob_e11}\\
&\mathsf{P}\{ E_{1,2} \} \leq K^{-2}.\label{eq:prob_e12}
\end{align}
in Appendices~\ref{append:proof_singleton0},~\ref{append:proof_singleton1}, and ~\ref{append:proof_singleton2}, respectively.}


We thus conclude 
\begin{align}\label{eq:prob_e1}
\mathsf{P} \left\{ E_{1} \right\} \leq 4 K^{-2}.
\end{align}

\subsubsection{Multiton Error Detection} It is proved in Appendix~\ref{append:proof_multiton} that 
\begin{align}\label{eq:prob_e2}
\mathsf{P} \left\{ E_{2} \right\} \leq 2 K^{-2}.
\end{align}
Therefore, combining \eqref{eq:prob_e0}, \eqref{eq:prob_e1} and \eqref{eq:prob_e2}, we conclude that the robust subcarrier detection makes the correct decision with probability higher than $1- 7  K^{-2}$.

We have \DG{thus} established Propositions~\ref{claim:step1}-\ref{claim:step3}.

\subsection{Complexity}

We perform $B$-point DFT for $C$ symbols. The computational complexity is \Lina{$O( K (\log K) (\log N + \log S + \log K) )$} using FFT operations. In robust subcarrier detection implemented after initialization of Algorithm~\ref{algo:suuc_cancel}, for each subcarrier, the phase estimation involves $O(\log N)$ operations, the device identification and message decoding involves $O(\log N + \log S)$ operations, \Lina{and the singleton detection and verification involves $O(\log K)$ operations}. The \textit{while} loop of Algorithm~\ref{algo:suuc_cancel} will be implemented no more than $K$ times. In each \textit{while} loop, delay estimation is neglected when $M=0$. For every subcarrier $b' \in \mathcal{S}$, where $|\mathcal{S}| \leq T$, the cancellation of signal of the detected device involves \Lina{$O(\log N + \log S + \log K)$} operations. The robust subcarrier detection for subcarrier $b'$ invloves \Lina{$O(\log N + \log S + \log K)$} operations. Since $T$ is a constant, the \textit{while} loop involves no more than \Lina{$O(K (\log N + \log S + \log K))$} operations. As a result, the complexity of DFT dominates, leading to the total computational complexity of \Lina{$O \left( K (\log K) (\log N + \log S + \log K) \right)$}. 



\DG{Hence the proof of} Theorem~\ref{theorem:noisy_discovery_sync}.

\section{Proof of Theorem \ref{theorem:noisy_discovery_asyn} (the Asynchronous Case)}
\label{sec:proof_asyn}
\XC{
In the asynchronous massive access case, we choose the parameters according to \eqref{eq:para_start}--\eqref{eq:para_end}. The last subframe in Fig.~\ref{fig:frame} consists of \Lina{$C_3 = M+\lceil\beta_3 K \log (KM+1)\rceil $} samples, where $\LL{\beta_3} \geq 64 \bar{a}^2 / a^2$. 
The total codelength in transmitted symbols is thus 
\Lina{\begin{align}
L  =& (B + M) C+ C_3 \\
= &(\beta_0K+M)\big((1+\lceil 1/R\rceil)\lceil\log N\rceil+\lceil 1/R\rceil\lceil\log S\rceil + \beta_1\lceil\log K\rceil\big)+M+ \lceil\beta_3K\log (KM+1)\rceil . 
\end{align}
}\DG{Thus the codelength satisfies~\eqref{eq:Lasync} for some $\alpha_0>0$.}
The FFT operation and channel estimation involve the same number of operations as the synchronous case. Different from the synchronous case, each device needs to estimate its delay once. The complexity of delay estimation is $O(M (M + K \log (KM+1))$ (corresponding to $M$ times auto-correlations). A total of $K$ devices need to estimate their delays. The total computational complexity is thus \Lina{$O \left( K (\log K) (\log N + \log S + \log K) \right) + O \left( K M ^2   + K^2 M \log (KM+1) \right)$}.

The following lemma shows that for each device the delay estimate is correct with high probability.
}
\begin{lemma}\label{lemma:delay_est}
Suppose the conditions specified in Theorem~\ref{theorem:noisy_discovery_asyn} hold. Suppose the bipartite graph $G \in \mathcal{G}$ and the parameters are chosen according to \eqref{eq:para_start}--\eqref{eq:para_end}. Suppose $\LL{\beta_3} \geq 64 \bar{a}^2 / a^2$ and is chosen to be an integer, and $C_3 = M + \Lina{\lceil\beta_3 K \log (KM+1)\rceil} $ samples are used for delay estimation, the delay of a device estimated according to~\eqref{eq:est_delay} is correct with probability no less than $1-4/K^2$.
\end{lemma}

The proof of Lemma~\ref{lemma:delay_est} is relegated to Appendix~\ref{append:proof_delay_est}.

Denote by $\mathcal{V}$ the event that the delay estimation of all devices is correct. Then we have $\mathsf{P} \{\mathcal{V} | G \in \mathcal{G} \} \geq 1 - 4/K$ by the union bound and Lemma~\ref{lemma:delay_est}. Conditioned on that the device delays are correctly detected, the residual errors of channel estimation can be characterized in the same manner as the synchronous case. Thus, the error probability can be upper bounded by 
\begin{align}\label{eq:err_prob}
\LL{P_s\leq \mathsf{P} \{ \mathcal{E} | G \in \mathcal{G}\} + \mathsf{P} \{ G \notin \mathcal{G}|G\in\hat{\mathcal{G}}  \}+\mathsf{P}\{G\notin\hat{\mathcal{G}}\},}
\end{align}
where
\begin{align}
     \mathsf{P} \{ \mathcal{E} | G \in \mathcal{G} \} &\leq 
     \mathsf{P} \{\bar{\mathcal{V}} | G \in \mathcal{G} \} + \mathsf{P} \{ \mathcal{E} |\mathcal{V}, G \in \mathcal{G} \}  \mathsf{P} \left\{ \mathcal{V} | G \in \mathcal{G} \right\} \\
     & \leq \frac{4}{K} + \mathsf{P} \{ \mathcal{E} |\mathcal{V}, G \in \mathcal{G} \} \\
     & \leq \frac{4}{K} + 1 - \left( 1 - \frac{7}{K^2} \right)^{KT} \label{eq:app_prop4}\\
     & \leq \frac{4}{K} + \frac{7 T}{K},\label{eq:theo2_epsilon}
\end{align}
where \eqref{eq:app_prop4} is according to Proposition~\ref{claim:step3}. \LL{We thus have
\begin{align}
 P_s \leq \frac{4 + 7 T + \nu +\zeta}{K}
\end{align}
according to Proposition~\ref{claim:step1}, \eqref{eq:err_prob}, and \eqref{eq:theo2_epsilon}.}

Theorem~\ref{theorem:noisy_discovery_asyn} is proved with $K_0 = \max\{(4 + 7 T + \LL{\nu+\zeta}) / \epsilon,K_0'\}$, and $\alpha_1$ being some constant due to the FFT operation.

\section{Simulation Results}
\label{sec:sim}

\LL{Without loss of generality}, we focus on the performance of device identification via sparse OFDMA throughout all simulations. Specifically, the total number of devices is $N = 2^{38}$, \LL{which is over 274 billion. (Alternatively, one can have one million devices where each active device can transmit $\log(274000)\approx 18$ bits)}. The performance of sparse OFDMA will be compared with two random acces schemes, namely slotted ALOHA and CSMA. Throughout the simulation, the channel coefficient amplitude is uniformly randomly generated from $[1, 2]$ and the phase is uniformly randomly generated from $[0, 2 \pi]$. The received SNR is defined as $\text{SNR} = 10\log( 1/ (2 \sigma^2)  )$, where $\sigma^2$ is the noise variance per dimension.

\subsection{Synchronous device identification}
\begin{figure}
\centering
  \includegraphics[width=10cm]{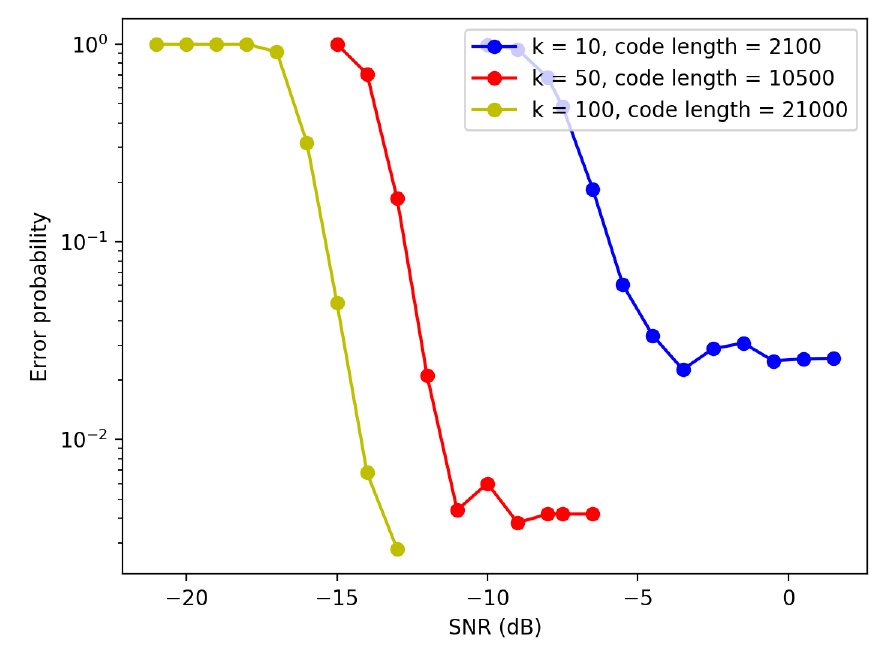}\\
  \caption{Error probability of device identification in the case of synchronous transmission, where $N = 2^{38}$.}\label{fig:sync_supp}
\end{figure}

\begin{figure}
\centering
  \includegraphics[width=10cm]{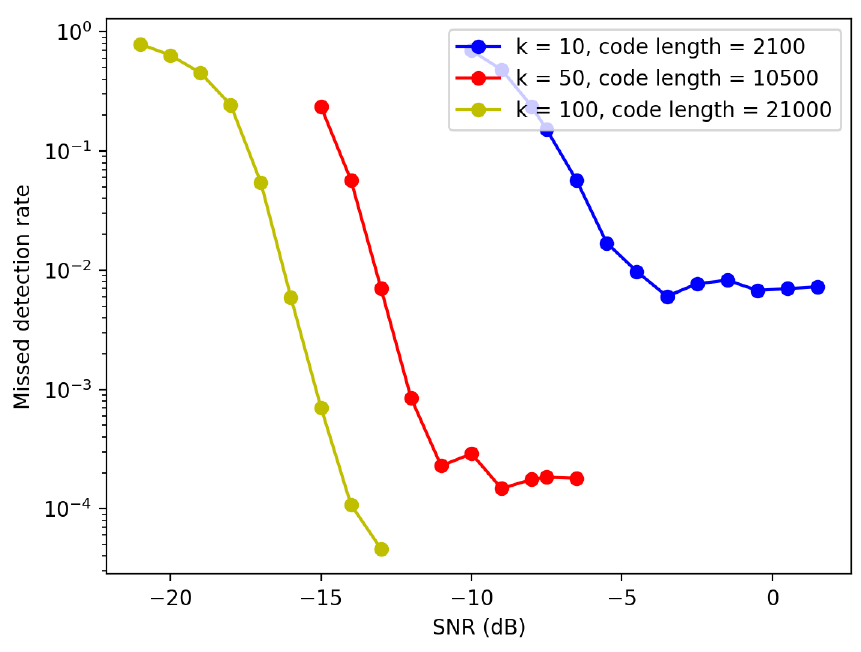}\\
  \caption{Rate of missed detection in the case of synchronous transmission, where is $N = 2^{38}$.}\label{fig:sync_md}
\end{figure}


We first investigate the error probability of synchronous device identification via sparse OFDMA. We choose the parameters as $B = \lceil 2.5K \rceil$, $C_0 = C_2 = 4$. The code rate $R=0.5$ is used to encode the device index, and thus $C_1 = 2  \lceil  \log N \rceil$. The number of subcarriers assigned to each device is $T = 3$.

Fig.~\ref{fig:sync_supp} shows the error probability of synchronous device identification. In each simulation, if there exists missed detection or false alarm, an error is registered. Throughout the simulation, the false positive rates is smaller than $10^{-5}$ for the SNRs in the region of interest, and thus we do not plot the results here. 
 Fig.~\ref{fig:sync_md} shows the miss rate, defined as the average number of misses normalized by the number of active devices $K$. Simulation shows that with a code length of $10500$, sparse OFDM can achieve a miss detection rate lower than $10^{-3}$ at SNR=$-10$ dB. In the case of a 20 MHz channel bandwidth, the transmission time is approximately 0.5 ms. 

\subsection{Asynchronous device identification}

\begin{figure}
\centering
  \includegraphics[width=10cm]{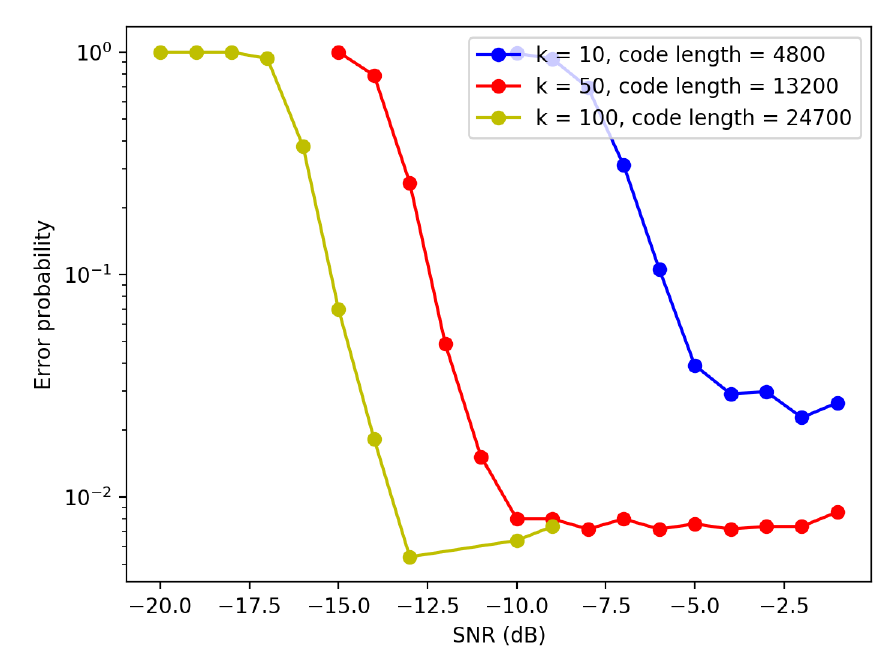}\\
  \caption{Error probability of device identification in the case of discrete delay. The device population is $N = 2^{38}$ and the maximum delay is $M= 20$.}\label{fig:discrete_supp}
\end{figure}

\begin{figure}
\centering
  \includegraphics[width=10cm]{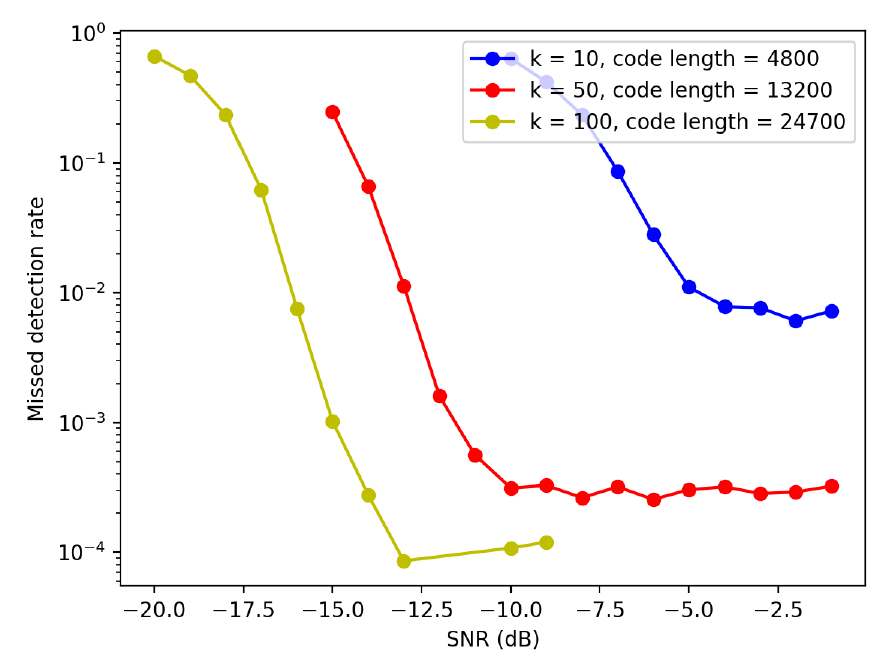}\\
  \caption{Rate of missed detection in the case of discrete delay. The device population is $N = 2^{38}$ and the maximum delay is $M= 20$.}\label{fig:discrete_md}
\end{figure}


We next investigate the error probability of asynchronous device identification. We choose the parameters as $B = \lceil 2.5 K \rceil$, $C_0 = C_2 = 4$. The code rate $R=0.5$ is used to encode the device index, and thus $C_1 = 2  \lceil  \log N \rceil$. The number of time-domain samples used for delay estimation is $C_3 = 1000$ for both cases of $K=10$ and $K=50$, and $C_3 = 2000$ for $K=100$. The device population is $N = 2^{38}$ and the maximum delay in terms of transmit samples is $M = 20$. Fig.~\ref{fig:discrete_supp} shows the error probability of asynchronous device identification. Fig.~\ref{fig:discrete_md} shows the missed detection rate (per active user), respectively. As in the synchronous setting, the error probability is low under moderate SNR, which confirms our theoretical analysis. 

\subsection{Comparison with random access}


\subsubsection{Slotted ALOHA}

First, consider slotted ALOHA, 
where every device transmits a frame with probability $p$
independently in each slot over an $N_s$-slot period.
The probability of one given neighbor being missed is equal to the probability that the device is
unsuccessful in all $N_s$ slots:
\begin{align}
\mathsf{P}_{\text{miss,aloha}} = \left(1 - (1 - p)^{K-1} p \right)^{N_s}.
\end{align}
Setting $p = 1/ K$ minimizes $\mathsf{P}_{\text{miss,aloha}}$.

\subsubsection{CSMA}

It is challenging, if not impossible, to implement CSMA-based
wireless access. Due to the power asymmetry between
devices and access points, a device may not be able to sense another
device's transmission in the same cell. Suppose, nonetheless, devices
can sense each other and CSMA is used.
When the channel is idle, the devices start their timers.
The device whose timer expires the first transmits.
When the channel becomes busy, the devices stop their timers.
Device $i$ has a chance to transmit if its timer is
the minimum in some slot.
The probability that a given device never gets a chance to transmit is
\begin{align}
\mathsf{P}_{\text{miss,csma}} = \left( 1 - \mathsf{P} \left\{ T_1 < \min\{ T_2, \cdots, T_K \} \right\} \right)^{N_s}.
\end{align}

In order to reliably transmit the device index $\log N$ bits, the number of symbols required in each frame is at least $\lceil \log N \rceil / \log(1+ \text{SNR})$.  Therefore, the total number of symbols required is $N_s \lceil \log N \rceil / \log(1+\text{SNR})$, where $N_s$ depends on the target miss rate.
Under SNR = -10 dB, it can be seen from Fig.~\ref{fig:discrete_md} that sparse OFDMA can achieve missed detection rate low than $10^{-3}$ for $K=10, 50$ or 100. We can calculate the number of symbols required by slotted ALOHA and CSMA to achieve a comparable miss detection rate of $10^{-3}$ at SNR = -10 dB. When $K=50$, the code length of sparse OFDMA is around 10,500, while slotted ALOHA and CSMA requires more than 90,000 symbols to achieve a missed detection rate of $10^{-3}$. Sparse OFDMA can effectively reduce the code length by over 80\%. Moreover, the code length reduction is even greater for larger $K$ and a lower error probability requirement. When $K=100$, the code length of sparse OFDMA is around 21,000, while slotted ALOHA and CSMA requires more than 180,000 symbols to achieve a missed detection rate of $10^{-3}$. Sparse OFDMA can effectively reduce the code length by over 85\%. The significant reduction of the code length required by sparse OFDM is due to the code design, which utilizes the sparsity of the active users among the total number of users.

\section{Conclusion}
\label{sec:conclude}

We have proposed a low-complexity asynchronous neighbor discovery and massive access scheme for very large networks with billions of nodes. The scheme may find important applications in the Internet of Things. The scheme, referred to as sparse OFDMA, applies the recently developed sparse Fourier transform to compressed device identification. Compared with random access schemes, sparse OFDMA requires much shorter code length by exploiting the multiaccess nature of the channel and the multiuser detection gain. Sparse OFDMA is a divide-and-conquer approach of low complexity, where only point-to-point capacity approaching codes are adopted. It provides practical physical layer capability for multipacket reception.

\appendices

\section{Auxiliary Results on Sub-Gaussian Variables}
\label{append:conc_ineq}

We introduce the definition of sub-Gaussian variables, which will be used in the proof of the main theorems.

\begin{definition}\label{def:subgauss}
$X$ is $\sigma$-subGaussian if there exists $\sigma >0$ such that
\begin{equation}
\mathsf{E} \left\{ \exp(tX) \right\}	 \leq \exp (\sigma^2 t^2/2), \quad \forall t \geq 0.
\end{equation}
\end{definition}

\begin{definition}[subGaussian norm]\label{def:subgaussian_norm}
The subGaussian norm of the random variable $X$ is defined as
\begin{align}\label{eq:subgauss_norm}
||X||_{\phi_2} = \sup_{p \geq 1} p^{-1/2} \left( \mathsf{E} |X|^p \right)^{1/p}.
\end{align}
\end{definition}

\LL{The following three lemmas, which are established in~\cite{rivasplata2012subgaussian}, will be used in the proof.}

\begin{lemma}\label{lemma:subgaussian_scale}
Suppose $X$ is $\sigma$-subGaussian, then $aX$ is $a \sigma$-subGaussian.
\end{lemma}

\begin{lemma}\label{lemma:subgaussian_sum}
Suppose $X_1$ is $\sigma_1$-subGaussian, $X_2$ is $\sigma_2$-subGaussian. Moreover, they are independent. Then $X_1+X_2$ is $\sqrt{\sigma_1^2 + \sigma_2^2}$-subGaussian.
\end{lemma}

\begin{lemma}\label{lemma:rademacher_basic}
If a random variable $X$ is $\sigma$-subGaussian with zero mean, then for any $t > 0$, the following holds
\begin{align}
\mathsf{P} \{ X > t \} &\leq e^{-\frac{t^2}{2 \sigma^2}} \\
\mathsf{P} \{ X < -t \} &\leq e^{-\frac{t^2}{2 \sigma^2}}.
\end{align}
\end{lemma}

The following theorem characterizes the properties of subGaussian variables~\cite{rivasplata2012subgaussian}.
\begin{theorem}[Characterization of subGaussian variables]\label{theorem:subgaussian_property}
Let $\mathsf{E} X = 0$. The following are equivalent:
\begin{enumerate}
\item $\mathsf{E} (e^{tX}) \leq e^{ \frac{t^2}{4}}$ .
\item $\forall p \geq 1$, $\left( \mathsf{E} |X|^p \right)^{1/p} \leq \sqrt{2 p}$.
\end{enumerate}
\end{theorem}
%
%

The following theorem is on the concentration of subGaussian random variables.

\begin{theorem}[Hanson-Wright inequality~\cite{rudelson2013hanson}]\label{theorem:hanson_wright}
Let $\bZ = (Z_1, \cdots, Z_n) \in \mathbb{R}^n$ be a random vector with independent components $Z_i$ which satisfy $\mathsf{E} Z_i =0$ and the subGaussian norm $||Z_i||_{\phi_2} \leq K$. Let $\bA$ be an $n \times n$ matrix. Then, for every $t \geq 0$,
\begin{align}
& \mathsf{P} \left\{ |\bZ^T \bA \bZ - \mathsf{E} \left\{ \bZ^T \bA \bZ \right\}| > t \right\} \leq  2 \exp \left[ -\LL{\varsigma} \min \left( \frac{t^2}{K^4 || \bA ||_{F}^2}, \frac{t}{K^2 ||\bA||} \right) \right],
\end{align}
where the operator norm of $\bA$ is $||\bA|| = \max_{\bx \neq 0} \frac{||\bA \bx||_2}{||\bx||_2}$, the Frobenius norm of $\bA$ is $||\bA||_F = (\sum_{i,j} |A_{i,j}|^2)^{1/2}$ and $\LL{\varsigma}$ is \LL{an absolute constant that does not depend on $K,\bA$ and $t$}.
\end{theorem}

%

\section{\LL{Proof of Singleton Error $E_{1,0}$ \eqref{eq:prob_e10}}}
\label{append:proof_singleton0}

We first establish the following lemma, which is key to the proof of \eqref{eq:prob_e10}.

\begin{lemma}\label{lemma:concentrate}
Let $B$ be given by \eqref{eq:para_B}.
Let $0 \leq c < C_2$ and $\LL{\beta_2} \geq 0$ be some fixed constants. Let each entry of $\bZ \in \mathbb{C}^{C_2}$ be distributed according to $\mathcal{CN}(0, 2 \sigma_z^2)$, where $\sigma_z \leq (1 + \LL{\beta_2}/2) \sigma^2 / B$. Let $\bQ = \bU \bLambda \bU^{\dagger}$, where $\bU$ is a real-valued orthogonal matrix and $\bLambda = \text{diag} \{0, \cdots ,0, 1, \cdots, 1 \}$ is a diagonal matrix with $c$ zeros and $C_2 - c$ ones. Let 
\begin{align}\label{eq:S}
\bS = \sum_{k=1}^{k_0} A_{k,b} \dot{\bg}_k,
\end{align} 
where $k_0 \geq 1$,  the entries of $\dot{\bg}_k$ are i.i.d. Rademacher variables, and $|A_{k,b}| \geq a$. Then there exists some $\beta_1$ such that the following holds for large enough $K$:
\begin{align}
\mathsf{P} \left\{ || \bQ \bS + \bZ||_2^2 \leq \eta \right\} \leq \frac{2}{K^2},\label{eq:lemma6}
\end{align}
where $C_2$ depends on $\beta_1$ according to  \eqref{eq:para_end} and $\eta$ is the energy threshold.
\end{lemma}

\begin{proof}
In the following, we use $\sum_{k} $ as a shorthand for $\sum_{k=1}^{k_0} $. We have
\begin{align}
\nonumber  \mathsf{P} \left\{ || \bQ \bS + \bZ||_2^2 \leq \eta \right\}  \leq &  \mathsf{P} \left\{ || \bQ \bS + \bZ||_2^2 \leq  \eta \bigg| ||\bQ \bS||_2^2 \geq \vartheta \right\}+ \\
&  \mathsf{P} \left\{ ||\bQ \bS||_2^2 \leq \vartheta \right\},
\label{eq:ub_multi_fa}
\end{align}
\LL{where the threshold $\vartheta$ is defined as 
\begin{align}
    \vartheta=\frac12 (C_2-c) \sum_{k} |A_{k,b}|^2.
\end{align}}
\DG{(Note $(C_2-c)\sum_k|A_{k,b}|^2$ is the expectation of $\Vert\bQ\bS\Vert_2^2$, hence the choice of threshold as a half of the expectation.)}


\LL{To bound the first term on the right hand side of~\eqref{eq:ub_multi_fa}, we use the triangular inequality $|| \bZ ||_2 + || \bQ \bS + \bZ ||_2 \geq || \bQ \bS ||_2 $ to obtain}
\begin{align}
\nonumber& \mathsf{P} \left\{ || \bQ \bS + \bZ ||_2^2 \leq  \eta \bigg| ||\bQ \bS||_2^2 \geq \vartheta \right\} \\
& \leq \mathsf{P} \left\{ || \bQ \bS ||_2 -  || \bZ ||_2 \leq  \sqrt{\eta} \bigg| || \bQ \bS ||_2^2 \geq \vartheta \right\}.
\end{align}
\Lina{For large enough $K$, $\eta$ can be chosen to be an arbitrarily small constant to satisfy \eqref{eq:eta}. Recall $C_2=\lceil\beta_1\log K\rceil$ as given by \eqref{eq:para_end}, for any fixed $c$, we have $\vartheta \geq \frac12 (C_2-c) a^2$, which is greater than $ 4 \eta$ for large enough $K$.} Therefore, for large enough $K$,
\begin{align}
\nonumber&\mathsf{P} \left\{ || \bQ \bS ||_2 -  || \bZ ||_2 \leq  \sqrt{\eta} \bigg| || \bQ \bS||_2^2 \geq \vartheta \right\} \\
& \leq \mathsf{P} \left\{  || \bZ ||_2 \geq  \sqrt{\eta} \bigg| || \bQ \bS||_2^2 \geq \vartheta \right\} \\
\label{eq:fa_ub_2} &\leq 1/K^2,
\end{align}
where \eqref{eq:fa_ub_2} \LL{follows from \eqref{eq:prob_cond_g_step1} and \eqref{eq:prob_e0}.}


The second term on the right hand side of~\eqref{eq:ub_multi_fa} is derived in the following steps. We first show that the real and imaginary parts of $\bS$ given by \eqref{eq:S} are subGaussian variables. Then, we show that $||\bQ \bS||^2_2$ are concentrated around $(C_2 -c) \sum_{k} |A_{k,b}|^2$ with high probability using the large deviation results on subGaussian variables. In the following, we denote $X_R$ and $X_I$ as the real and imaginary component of $X$, respectively.

\begin{lemma}
Let $\bS_R = (S_{0,R}, \cdots, S_{C_2-1,R})$ and $\bS_I = (S_{0,I}, \cdots, S_{C_2-1,I})$ be the real and imaginary components of $\bS$, respectively. Then $S_{c,R}$ are i.i.d. $\sqrt{\sum\limits_{k} A^2_{k,R}}$-subGaussian random variables with $\mathsf{E} S_{c,R} = 0$ and the subGaussian norm satisfies $|| S_{c,R} ||_{\phi_2} \leq 2 \sqrt{ \sum\limits_{k} A^2_{k,R}}$. Similarly,
$S_{c,I}$ are i.i.d. $\sqrt{\sum\limits_{k} A^2_{k,I}}$-subGaussian random variables with $\mathsf{E} S_{c,I} = 0$ and the subGaussian norm satisfies $|| S_{c,I} ||_{\phi_2} \leq 2 \sqrt{ \sum\limits_{k} A^2_{k,I}}$.
\end{lemma}
\begin{proof}
Since $\dot{g}^c_{k}$ is Rademacher variable, it is 1-subGaussian with mean zero. Thus, $\mathsf{E} S_{c,R} = 0$. Moreover, $\dot{g}^c_{k}$ are independent across $k$. According to Lemma~\ref{lemma:subgaussian_scale} and Lemma~\ref{lemma:subgaussian_sum}, $S_{c,R} = \sum_{k} A_{k,R} \dot{g}_{k}^c$ is $\sqrt{\sum\limits_{k} A^2_{k,R}}$-subGaussian. $\{S_{c,R} \}_{c=0}^{C_2 -1}$ are independent, because $\dot{g}_{k}^c$ are independent across $c = 0, \cdots, C_2-1$.

By Definition~\ref{def:subgauss}, $\mathsf{E} \left\{ \exp( t S_{c,R} ) \right\} \leq \exp \left(\sum\limits_{k} A^2_{k,R} t^2 /2 \right)$. Let $X =  S_{c,R} / \sqrt{2 \sum\limits_{k} A^2_{k,R}}$. Then $\mathsf{E} \left\{ \exp( t X ) \right\} \leq \exp(t^2/4)$. According to Theorem~\ref{theorem:subgaussian_property}, for all $p \geq 1$, $(\mathsf{E} |X|^p)^{1/p} \leq \sqrt{2p}$, which yields
\begin{align}
\left( \mathsf{E} |S_{c,R}|^p \right)^{1/p} \leq 2 \sqrt{ \sum\limits_{k} A^2_{k,R} p}.
\end{align}
By \eqref{eq:subgauss_norm}, the subGaussian norm of $S_{c,R}$ is upper bounded as $|| S_{c,R} ||_{\phi_2} \leq 2 \sqrt{ \sum\limits_{k} A^2_{k,R}}$. The statement for the imaginary parts follow similarly. \qed
\end{proof}

In order to apply Theorem~\ref{theorem:hanson_wright} to provide a concentration result on $||\bQ \bS||$, we need to first derive the Frobenius norm and the operator norm of $\bQ$. The Frobenius norm of $\bQ$ is calculated as
\begin{align}
||\bQ||^2_F &= tr\{\bQ \bQ^{\dagger} \} \\
\label{eq:tr_q} &= tr\{ \bQ \} \\
\label{eq:tr_q_1}& = C_2-c,
\end{align}
where \eqref{eq:tr_q} follows by $\bQ \bQ^{\dagger} = \bQ$, and \eqref{eq:tr_q_1} follows because the sum of the eigenvalues of $\bQ$ is $C_2-c$.

Since the largest eigenvalue of $\bQ$ is 1, the operator norm of $\bQ$ is calculated as
\begin{align}
||\bQ || = \max_{\bx \neq 0} \frac{|| \bQ \bx||_2}{||\bx||_2} = 1.
\end{align}

Conditioned on $\bg_{k_0}$,
\begin{align}
\mathsf{E} \left\{ ||\bQ \bS_R||^2_2  \right\} &= \mathsf{E} \left\{ \bS_R^{\dagger} \bQ \bQ^{\dagger} \bS_R \right\} \\
& = \mathsf{E} \left\{ \bS_R^{\dagger} \bQ  \bS_R  \right\} \\
\label{eq:mean_qz} & = \mathsf{E} \{S_{c,R}^2\} tr\{\bQ\} \\
& = (C_2-c) \sum_{k} A^2_{k,R},
\end{align}
where \eqref{eq:mean_qz} follows because $\{S_{c,R}\}_{c=0}^{C_2-1}$ are i.i.d. distributed. Similarly, $\mathsf{E} \left\{ ||\bQ \bS_I||^2_2  \right\} = (C_2-c) \sum_{k} A^2_{k,I}$.

Applying Theorem~\ref{theorem:hanson_wright} with $\bZ = \bS_R$, $\bA = \bQ$ with $||\bQ|| = 1$, $||\bQ||^2_F = C_2 -c$ and $K = 2 \sqrt{ \sum\limits_{k} A^2_{k,R}}$ yields
\begin{align}
\nonumber & \mathsf{P} \left\{ \bigg| || \bQ \bS_R ||^2_2 - (C_2-c) \sum_{k} A^2_{k,R} \bigg| >t  \right\} \leq \\
& 2 \exp \left(- \LL{\varsigma} \min \left( \frac{t^2}{  16 (C_2-c) (\sum_{k} A^2_{k,R})^2 }, \frac{t}{4 \sum_{k} A^2_{k,R}} \right) \right),
\end{align}
\LL{where $\varsigma$ is a constant introduced in Theorem~\ref{theorem:hanson_wright} in Appendix~\ref{append:conc_ineq}.}

Letting $t = (C_2-c) \sum_{k} A_{k,R}^2 /2$, we have
\begin{align}\label{eq:cond_lb_qz}
\mathsf{P} \left\{ || \bQ \bS_R||_2^2 \leq \frac{(C_2-c) \sum_{k} A^2_{k,R} }{2}  \right\} \leq 2 \exp \left(- \frac{\LL{\varsigma}}{64} (C_2-c) \right).
\end{align}
Similarly, we have
\begin{align}
\label{eq:fa_ub_4}\mathsf{P} \left\{ || \bQ \bS_I||_2^2 \leq \frac{(C_2-c) \sum_{k} A^2_{k,I} }{2} \right\} \leq 2 \exp \left(- \frac{\LL{\varsigma}}{64} (C_2-c) \right).
\end{align}

Since $\bQ$ is a real-valued matrix, $||\bQ \bS||_2^2 = ||\bQ \bS_R||_2^2 + ||\bQ \bS_I||_2^2$. Moreover, $\sum_k |A_{k,b}|^2 = \sum_k \left( A_{k,R}^2 + A_{k,I}^2 \right)$. Combining \eqref{eq:cond_lb_qz} and \eqref{eq:fa_ub_4}, we have
\begin{align}
\mathsf{P} \left\{ || \bQ \bS||_2^2 \leq \frac{(C_2-c) \sum_{k} |A_{k,b}|^2 }{2} \right\}  &\leq \mathsf{P} \left\{ || \bQ \bS_R||_2^2 \leq \frac{(C_2-c) \sum_{k} A^2_{k,R} }{2} \right\}  + \mathsf{P} \left\{ || \bQ \bS_I||_2^2 \leq \frac{(C_2-c) \sum_{k} A^2_{k,I} }{2} \right\}  \\
\label{eq:fa_ub_5} & \leq 4 \exp \left(- \frac{\LL{\varsigma}}{64} (C_2-c) \right).
\end{align}
Combining \eqref{eq:ub_multi_fa}, \eqref{eq:fa_ub_2} and \eqref{eq:fa_ub_5}, there exists some large enough $\beta_1$ such that \Lina{for $C_2 =   \lceil \beta_1\log K \rceil$},
\begin{align}
\mathsf{P}  \left\{ || \bQ \bS + \bZ||_2^2 \leq \eta \right\}  \leq \frac{2}{K^2}.
\end{align} \qed
\end{proof}

The probability that a singleton is declared to be a zeroton is calculated as 
\begin{align}
\mathsf{P} \left\{ E_{1,0} \right\}  = \mathsf{P} \left\{ || A_{k,b} \dot{\bg}_k + \dot{\bZ}_b ||_2^2 \leq \eta \right\}. 
\end{align}
Therefore, we can apply Lemma~\ref{lemma:concentrate} with $\bS = A_{k,b} \dot{\bg}_k$, $\bQ = \bI$, $\bZ = \dot{\bZ}_b $, $k_0 = 1$, and $c=0$ to obtain \eqref{eq:prob_e10}. In
this case, \Lina{$\beta_1$ can be chosen to satisfy $\beta_1\geq 192/\varsigma$} such that $4\exp\left(-\frac{\LL{\varsigma}}{64}\lceil\beta_1\log N\rceil\right)\leq K^{-2}$.

\section{Proof of Singleton Error $E_{1,1}$ \eqref{eq:prob_e11}
\label{append:proof_singleton1}}

We first show that the phase compensation is accurate with high probability. Second, we show that the interference only causes a slight degradation of signal strength with high probability. Third, the device index recovery can be regarded as transmission over a BSC channel and a large enough $C_1$ can help recover the index information.

We estimate the phase of $\theta= \angle A_{k,b}$ according to \eqref{eq:phase}. Then the estimated phase is calculated as
\begin{align}
\hat{\theta} & = \angle \left( A_{k,b} + \bar{Z} \right),
\end{align}
where $\bar{Z} =   \sum\limits_{c=0}^{C_0 - 1} \left( \bar{W}^c_{ b } + \bar{V}^c_{ b } \right) / C_0$ is Gaussian distributed with zero mean and variance $2 \sigma_{\bar{z}}^2$ conditioned on the design parameters, where $\sigma_{\bar{z}}^2 \leq (1+\LL{\beta_2}/2)\frac{\sigma^2}{BC_0}$. From the geometric interpretation, the maximum phase offsets occurs when the noise is orthogonal to the measurement. We choose a small $\theta_0$ such that $\theta_0 < \frac{\pi}{3}$ and $\sin \theta_0 > \theta_0 /2$, then
\begin{align}
\mathsf{P} \left\{ |\hat{\theta} - \theta| > \theta_0  \right\} &\leq \mathsf{E} \left\{ \mathsf{P} \left\{  \arcsin \frac{|\bar{Z}|}{|A_{k,b}|} > \theta_0 \right\}  | \bg \right\}\\
& \leq \mathsf{E} \left\{  \mathsf{P} \left\{  |\bar{Z}| > a \sin \theta_0  \right\}  | \bg \right\} \\
& \leq \mathsf{E} \left\{  \mathsf{P} \left\{  |\bar{Z}| > \frac{a \theta_0}{2}  \right\}  | \bg \right\} \\
& \leq \mathsf{E} \left\{  2 \mathsf{P} \left\{  |\text{Re} \{\bar{Z}\}| > \frac{a \theta_0}{ 2 \sqrt{2}}  \right\}  | \bg \right\} \\
\label{eq:phase_est_step1}& \leq 4 \exp \left(- \frac{a^2 \theta_0^2}{16 \sigma_{\bar{z}}^2 } \right) ,
\end{align}
where \eqref{eq:phase_est_step1} is due to $\mathsf{P} \{ X> x\} \leq e^{- x^2/2}$ with $X \sim \mathcal{N}(0,1) $ being the standard Gaussian random variable. For $\vert\hat{\theta}-\theta\vert \leq \theta_0$,
\begin{align}
\text{Re}\left\{A_{k,b}e^{-\iota\hat{\theta}}\right\}\geq& a \cos\left(\hat{\theta}-\theta\right)\\
\geq& a \cos\theta_0\\
\label{eq:phase_est}\geq& a /2.
\end{align}

Therefore, given $B = \beta_0 K$ and $C_0 = \lceil \log N \rceil$, we have
\begin{equation}\label{eq:Preal}
\mathsf{P}\left\{\text{Re}\left\{A_{k,b}e^{-\iota\hat{\theta}}\right\}\geq a/2\right\}\geq 1-4\exp\left(-\frac{a^2\theta_0^2\beta_0 K\log N}{(16+8\LL{\beta_2})\sigma^2}\right).
\end{equation}

To predict $\tilde{g}_k^c$, we make (hard) binary decision on
\begin{align} \text{Re}\left\{\tilde{Y}_b^c e^{-\iota\hat{\theta}}\right\}&= \text{Re}\left\{(A_{k,b}\tilde{g}_k^c+\tilde{V}_b^c+\tilde{W}_b^c)e^{-\iota\hat{\theta}}\right\}\\
&= \text{Re}\left\{A_{k,b}\tilde{g}_k^c e^{-\iota\hat{\theta}}\right\}+\text{Re}\left\{\tilde{Z}_b^c e^{-\iota\hat{\theta}} \right\},
\end{align}
where $\tilde{Z}_b^c=\tilde{V}_b^c+\tilde{W}_b^c$ is distributed according to $\mathcal{CN}(0,2\sigma_{\tilde{z}}^2)$ with $\sigma_{\tilde{z}}^2\leq (1+\LL{\beta_2}/2)\sigma^2/B$, conditioned on the design parameters. We consider the corruption of signal strength from interference plus noise. Thus, the prediction error for $\tilde{g}_k^c$ occurs when $\tilde{g}_k^c\text{Re}\{\tilde{Z}_b^ce^{-\iota\hat{\theta}}\}<-\text{Re}\{A_{k,b}e^{-\iota\hat{\theta}}\}$. Then the flip rate can be upper bounded by
\LL{
\begin{equation}
\mathsf{P}\left\{\text{Re}\left\{A_{k,b}e^{-\iota\hat{\theta}}\right\}<a/2\right\}+\mathsf{P}\left\{\text{Re}\left\{A_{k,b}e^{-\iota\hat{\theta}}\right\}\geq a/2,\tilde{g}_k^c\text{Re}\{\tilde{Z}_b^c e^{-\iota\hat{\theta}} \}<-\text{Re}\{A_{k,b}e^{-\iota\hat{\theta}}\}\right\},\label{eq:flip}
\end{equation}
where the first term
\begin{equation}
\mathsf{P}\left\{\text{Re}\left\{A_{k,b}e^{-\iota\hat{\theta}}\right\}<a/2\right\}\leq 4\exp\left(-\frac{a^2\theta_0^2\beta_0 K\log N}{(16+8\LL{\beta_2})\sigma^2}\right)\label{eq:flip_1}
\end{equation}
according to \eqref{eq:Preal}, and the second term 
\begin{align}
\mathsf{P}&\left\{\text{Re}\left\{A_{k,b}e^{-\iota\hat{\theta}}\right\}\geq a/2,\tilde{g}_k^c\text{Re}\{\tilde{Z}_b^c e^{-\iota\hat{\theta}} \}<-\text{Re}\{A_{k,b}e^{-\iota\hat{\theta}}\}\right\}\notag\\
&\leq\mathsf{P}\left\{\tilde{g}_k^c\text{Re}\{\tilde{Z}_b^c e^{-\iota\hat{\theta}} \}<-a/2\right\}\label{eq:A_real}\\
&\leq\mathsf{P}\left\{ -|\tilde{Z}_b^c| <-a/2\right\}\label{eq:z_gaussian_0}\\
&=\exp\left(-\frac{a^2}{8\sigma_{\tilde{z}}^2}\right)\label{eq:rayleigh}\\
&\leq\exp\left(-\frac{a^2\beta_0K}{(8+4\beta_2)\sigma^2}\right),\label{eq:flip_2}
\end{align}
where \eqref{eq:z_gaussian_0} is due to $\tilde{g}_k^c\text{Re}\{\tilde{Z}_b^c e^{-\iota\hat{\theta}} \} \geq - |\tilde{Z}_b^c|$, and \eqref{eq:rayleigh} is because $\vert\tilde{Z}_b^c\vert$ follows a Rayleigh distribution since $\tilde{Z}_b^c\sim\mathcal{CN}(0,2\sigma_{\tilde{z}}^2)$. 
Combining \eqref{eq:flip}, \eqref{eq:flip_1} and \eqref{eq:flip_2}, the flip rate can be upper bounded by $4\exp\left(-\frac{a^2\theta_0^2\beta_0 K\log N}{(16+8\LL{\beta_2})\sigma^2}\right)+\exp\left(-\frac{a^2\beta_0K}{(8+4\beta_2)\sigma^2}\right)$.} This indicates that the flip rate tends to $0$ when $K,N\rightarrow\infty$. 
\LL{By \cite[Theorem 6.1]{barg2004error}, there exists a code rate that depends only on the given flipping rate of the binary symmetric channel, to encode the  $(\lceil \log N \rceil + \lceil \log S \rceil)$-bit information, such that the index can be recovered correctly with probability at least $1 - 1/ N^2$. For large enough $K$, the flipping rate $e^{-O(K)}$ vanishes. Thus, incorrect device identification and message decoding occurs with probability 
\begin{align}
\mathsf{P} \left\{ E_{1,1} \right\} \leq N^{-2} \leq K^{-2}. 
\end{align}
with a certain fixed rate $R$.}

\section{Proof of Singleton Error $E_{1,2}$ \eqref{eq:prob_e12}
\label{append:proof_singleton2}}

Let $\dot{A}_{k,b} = \dot{\bg}_k^{\dagger} \dot{\bY}_b /C_2$. We have
\begin{align}
 \dot{\bY}_{b} - \dot{A}_{k,b} \dot{\bg}_{k} = \left( \bI - \frac{1}{C_2} \dot{\bg}_{k} \dot{\bg}_{k}^{\dagger} \right) \dot{\bZ}_b .
\end{align}
Let $\bQ =  \bI - \frac{1}{C_2} \dot{\bg}_{k} \dot{\bg}_{k}^{\dagger}$. Since $\dot{\bg}_{k}^{\dagger} \dot{\bg}_{k} = C_2$, $\dot{\bg}_{k}$ is the eigenvector of $ \dot{\bg}_{k} \dot{\bg}_{k}^{\dagger} / C_2$. Since $ \dot{\bg}_{k} \dot{\bg}_{k}^{\dagger} / C_2$ is a rank-1 matrix, it has only a single nonzero eigenvalue which is 1. Therefore, the eigenvalue decomposition of $\bQ$ can be written as
\begin{align}
 \bQ &= \bU \Lambda \bU^{\dagger}  \\
\label{eq:eigen_q}& =  \bU  {\rm diag} \{0, 1, \cdots, 1\} \bU^{\dagger}
\end{align}
where $\bU$ is an orthogonal matrix. Then we have
\begin{align}
||  \dot{\bY}_{b} - \dot{A}_{k,b} \dot{\bg}_{k} ||^2_2 &= || \bQ \dot{\bZ}_b ||_2^2 \\
&= || \Lambda \bU^{\dagger} \dot{\bZ}_b ||_2^2 \\
\label{eq:err_zeroton_1} &= \sum\limits_{c=1}^{C_2-1} |Z'_c|^2,
\end{align}
where $\bZ' = \bU^{\dagger} \dot{\bZ}_b $ and \eqref{eq:err_zeroton_1} is due to $\Lambda = {\rm diag} (0, 1, \cdots, 1)$. Since $||\dot{\bZ}_b ||_2^2 = || \bZ'||_2^2$, we have
\begin{align}
\mathsf{P} \left\{ ||  \dot{\bY}_{b} - \dot{A}_{k,b} \dot{\bg}_{k} ||^2_2 \geq \eta  \right\} &= \mathsf{P} \left\{  \sum\limits_{c=1}^{C_2-1} |Z'_c|^2 \geq \eta  \right\} \\
& \leq \mathsf{P} \left\{  ||\dot{\bZ}_b ||_2^2 \geq \eta  \right\} \label{eq:z_eta}\\
\label{eq:ub_Es3} & \leq K^{-2},
\end{align}
where \eqref{eq:ub_Es3} follows from \eqref{eq:prob_cond_g_step1} and \eqref{eq:prob_e0}. 


\section{Proof of Multiton Error \eqref{eq:prob_e2}}
\label{append:proof_multiton}

We rewrite the subcarrier values as follows,
\begin{align}
\dot{\bY}_b  & = \sum\limits_{ k \in \mathcal{K}: b \in \mathcal{B}_k   } A_{k,b} \dot{\bg}_{k} + \dot{\bZ}_{b},
\end{align}
where $\dot{\bZ}_b = \dot{\bW}_b + \dot{\bV}_b$.

Suppose the incorrect estimate index from subcarrier-$b$ is $k_b$. We have
\begin{align}
\dot{A}_{k_b} = \sum\limits_{ k \in \mathcal{K}: b \in  \mathcal{B}_k    } \frac{1}{C_2} A_{k,b} \dot{\bg}_{k_b}^{\dagger} \dot{\bg}_{k} + \frac{1}{C_2} \dot{\bg}_{k_b}^{\dagger} \dot{\bZ}_{b}.
\end{align}
Thus,
\begin{align}
 \label{eq:fa_test}\dot{\bY}_b - \dot{A}_{k_b} \dot{\bg}_{k_b} &= \sum\limits_{ k \in \mathcal{K}:  b \in  \mathcal{B}_k   } A_{k,b} \left( \bI - \frac{ \dot{\bg}_{k_b} \dot{\bg}_{k_b}^{\dagger}}{C_2} \right) \dot{\bg}_{k} +  \left( \bI - \frac{ \dot{\bg}_{k_b} \dot{\bg}_{k_b}^{\dagger}  }{C_2}  \right) \dot{\bZ}_{b}.
\end{align}

Let
\begin{align}\label{eq:sum_rademacher}
\bS =  \sum\limits_{ k \in \mathcal{K}: b \in  \mathcal{B}_k  } A_{k,b} \dot{\bg}_{k}
\end{align} and $\bQ =  \bI - \frac{1}{C_2} \dot{\bg}_{k_b} \dot{\bg}_{k_b}^{\dagger}$, then the first term in \eqref{eq:fa_test} can be written as
\begin{align}
\sum\limits_{ k \in \mathcal{K}:  b \in  \mathcal{B}_k    } A_{k,b} \left( \bI - \frac{ \dot{\bg}_{k_b} \dot{\bg}_{k_b}^{\dagger}}{C_2} \right) \dot{\bg}_{k} = \bQ \bS.
\end{align}

The multiton subcarrier cannot be detected when $|| \bY_b - \dot{A}_{k_b, b} \dot{\bg}_{k_b} ||_2^2 \leq \eta$. In the following, we upper bound the error probability assuming $k_b \notin \{ k \in \mathcal{K}: b \in \mathcal{B}_k \}$. A similar analysis can be carried out for the case of $k_b \in \{ k \in \mathcal{K}: b \in \mathcal{B}_k \}$. Conditioned on $\dot{\bg}_{k_b}$, we can apply Lemma~\ref{lemma:concentrate} with $c=1$, $k_0 = |k \in \mathcal{K}:  b \in  \mathcal{B}_k |$ and $\bZ =  \bQ \dot{\bZ}_b$ to upper bound the error probability as
\begin{align}\label{eq:multiton_tmp}
\mathsf{P} \left\{ || \bY_b - \dot{A}_{k_b,b} \dot{\bg}_{k_b} ||_2^2 \leq  \eta | \dot{\bg}_{k_b} \right\} &= 
\mathsf{P} \left\{ || \bQ \bS + \bQ \dot{\bZ}_b ||_2^2 \leq  \eta | \dot{\bg}_{k_b} \right\} \\
& \leq 2 K^{-2},\label{eq:lemma6_result}
\end{align}
\LL{where \eqref{eq:lemma6_result} directly follows from \eqref{eq:lemma6}.} Since \eqref{eq:multiton_tmp} holds for every      $\dot{\bg}_{k_b}$, we have 
\begin{align}
\mathsf{P} \left\{ || \bY_b - \dot{A}_{k_b,b} \dot{\bg}_{k_b} ||_2^2 \leq  \eta \right\} \leq 2 K^{-2}.
\end{align}


\section{Proof of Lemma~\ref{lemma:delay_est}}
\label{append:proof_delay_est}

We focus on the delay estimation for device $k$. Without loss of generality, we assume the delay is $m_k = 0$. The device experiences the interference from the other $K-1$ devices and noise. The received synchronization pilots can be written as 
\begin{align}
x_i = a_k s_{k, i} +\sum_{p \in \mathcal{K} \backslash k} a_{p} s_{p, i-m_p} + w_i, 
\end{align}
where the time-domain samples of the pilots $s_{k,i}$ are uniform random from $\{+1, -1\}$ and the noise $w_i \sim \mathcal{CN}(0, 2\sigma^2)$.

Let $\mathcal{I}$ denote the samples in the \LL{last} subframe with the first $M$ samples discarded. The number of samples contained in $\mathcal{I}$ is $|\mathcal{I}| = \Lina{\lceil\beta_3 K \log (KM+1)\rceil} $, $\LL{\beta_3} \geq 64 \bar{a}^2 / a^2$. Since the random sequence has a length greater than $M$ and the delay is less than $M$, without noise, correlating the second segment with the received sequence will yield an auto-correlation of a sequence with its cyclic shift version. According to \eqref{eq:metric}, the test metric is calculated as 
\begin{align}\label{eq:tm_metric}
\mathcal{T} (m) = 
\left\{
\begin{array}{cc}
 |\mathcal{I}| a_k + \sum_{i \in \mathcal{I} }  w_{i+M} s_{k,i+M} + \sum_{p \in  \mathcal{K} \backslash k} \sum_{i \in \mathcal{I} }  a_p s_{p, i+M-m_p} s_{k,i+M} & \text{ if }m = 0 \\
 \sum_{i \in \mathcal{I}} a_k s_{k,i+M+m} s_{k,i+M} +  \sum_{p \in  \mathcal{K} \backslash k} \sum_{i \in \mathcal{I}} a_p s_{p, i+M+m-m_p} s_{k,i+M}+ \sum_{i \in \mathcal{I}}w_{i+M+m} s_{k,i+M}  & \text{ if } m \neq 0. \\
\end{array} \right.
\end{align}
By the choice of the random sequence, when $p \neq k$, $s_{p, i+M+m-m_p} s^{\ast}_{k,i+M}$ are i.i.d. Rademacher variables, for all $p$ and $i \in \mathcal{I}$. Moreover, when $m \neq 0$, $s_{k,i+M+m} s_{k,i+M}$ are i.i.d. Rademacher variables for all $i \in \mathcal{I}$. We will show that 
with a large enough of $\LL{\beta_3}$, the error probability can be lower than $O(1/K^2)$.

Define a threshold $\mathcal{\bar{T}} = a  |\mathcal{I}| /2$. If $\text{Re} \{\mathcal{T} (0) \} > \mathcal{\bar{T}}$ and $\text{Re} \{\mathcal{T} (m) \} < \mathcal{\bar{T}}$ for all $m = 1, \cdots, M-1$, the delay can be correctly estimated. Therefore, the error probability can be upper bounded as  
\begin{align}
\label{eq:prob_delay_0} P_e &\leq \sum_{m=1}^{M-1} \mathsf{P} \{ \text{Re} \{ \mathcal{T} (m) \} \geq\mathcal{\bar{T}} \} + \mathsf{P} \{ \text{Re} \{ \mathcal{T} (0) \} \leq \mathcal{\bar{T}} \} .
\end{align}

Let $r_{k,i}$ be random i.i.d. Radamacher variables for all $k \in \mathcal{K}$ and $i \in \mathcal{I}$. Let $z_{i} = \text{Re} \{w_i\}$ be the real part of the random noise $w_i$. For $m=1,\cdots,M-1$, we have
\begin{align}
\label{eq:bound_tm}\mathsf{P} \left\{ \text{Re} \{ \mathcal{T} (m) \} \geq \mathcal{\bar{T}} \right\} &\leq \mathsf{P} \left\{  \sum_{p \in  \mathcal{K}} \sum_{i \in \mathcal{I}} a_p r_{p,i} \geq \mathcal{\bar{T}} /2 \right\} + \mathsf{P} \left\{ \sum_{i \in \mathcal{I}} z_{i+M+m} s_{k,i+M}   \geq \mathcal{\bar{T}} /2 \right\} .
\end{align}

We first derive the upper bound for the second term in \eqref{eq:bound_tm}. Conditioned on $s_{k,i+M}$, $k \in \mathcal{K}$ and $i \in \mathcal{I}$, the variables $z_{i+M+m} s_{k,i+M}  $ are i.i.d. Gaussian variables with mean zero and variance $\sigma^2$. Therefore, 
\begin{align}
\mathsf{P} \left\{ \sum_{i \in \mathcal{I}} z_{i+M+m} s_{k,i+M}  \geq \mathcal{\bar{T}} /2 \right\}  &= Q \left( \frac{\bar{T}}{2 \sqrt{|\mathcal{I}| \sigma^2} } \right) \\
& \leq e^{- \frac{ \bar{T}^2}{8 |\mathcal{I}| \sigma^2} } \\
\label{eq:ub_noise}& = e^{-\frac{a^2 |\mathcal{I}| }{32 \sigma^2}}.
\end{align}

We then derive the upper bound for the first term in \eqref{eq:bound_tm}.

Rademacher variables are $1$-subGaussian with mean zero. According to Lemma~\ref{lemma:subgaussian_scale} and Lemma~\ref{lemma:subgaussian_sum}, $\sum_p \sum_{i \in \mathcal{I}} a_p r_{p,i}$ is $\sqrt{\sum_p |\mathcal{I}| a_p^2}$-subGaussian variables with zero mean. Thus we have
\begin{align}
\mathsf{P} \left\{  \sum_{p \in \mathcal{K}} \sum_{i \in \mathcal{I}} a_p r_{p,i} \geq \mathcal{\bar{T}} / 2 \right\} 
& \leq  e^{-  \frac{\mathcal{\bar{T}}^2 }{8 |\mathcal{I}| \sum_p a^2_p}   } \\
\label{eq:ub_cross_terms}& \leq  e^{-   \frac{a^2 |\mathcal{I}| }{32  K \bar{a}^2 } }. 
\end{align}
Choosing $\LL{\beta_3} \geq 64 \bar{a}^2 / a^2$, \LL{\eqref{eq:ub_noise} and \eqref{eq:ub_cross_terms} are both upper bounded by $1/K^2$. Thus, by \eqref{eq:bound_tm}, we have }
\begin{align}
\sum_{m=1}^{M-1} \mathsf{P} \left\{ \text{Re} \{ \mathcal{T} (m) \} \geq \mathcal{\bar{T}} \right\} \leq \frac{2}{K^2}.
\end{align}

Following exactly the same derivations, we can obtain $ \mathsf{P} \{ \text{Re} \{ \mathcal{T} (0) \}  \leq 2 /K^2$. The detailed derivations is shown below. Given \eqref{eq:tm_metric} on $\mathcal{T}(m)$, we have
\begin{align}
\mathsf{P}\{\text{Re}\{\mathcal{T}(0)\} \leq \bar{\mathcal{T}}\}\leq&\mathsf{P}\left\{\frac{|\mathcal{I}|a_k}{2}+\sum_{p\in\mathcal{K}\backslash k}\sum_{i\in\mathcal{I}}a_pr_{p,i}\leq\bar{\mathcal{T}}/2\right\}+\mathsf{P}\left\{\frac{|\mathcal{I}|a_k}{2}+\sum_{i\in\mathcal{I}}z_{i+M}s_{k,i+M}\leq\bar{\mathcal{T}}/2\right\}\\
=&\mathsf{P}\left\{\sum_{p\in\mathcal{K}\backslash k}\sum_{i\in\mathcal{I}}a_pr_{p,i}\leq\bar{\mathcal{T}}/2-\frac{|\mathcal{I}|a_k}{2}\right\}+\mathsf{P}\left\{\sum_{i\in\mathcal{I}}z_{i+M}s_{k,i+M}\leq\bar{\mathcal{T}}/2-\frac{|\mathcal{I}|a_k}{2}\right\}\\
\label{eq:bound_tm_m0} \leq&\mathsf{P}\left\{\sum_{p\in\mathcal{K}\backslash k}\sum_{i\in\mathcal{I}}a_pr_{p,i}\leq-\bar{\mathcal{T}}/2 \right\}+\mathsf{P}\left\{\sum_{i\in\mathcal{I}}z_{i+M}s_{k,i+M}\leq-\bar{\mathcal{T}}/2\right\}.
\end{align}   
The first term can be upper bounded according to Lemma~\ref{lemma:rademacher_basic} as
\begin{align}
\label{eq:ub_cross_terms_2} \mathsf{P}\left\{\sum_{p\in\mathcal{K}\backslash k}\sum_{i\in\mathcal{I}}a_pr_{p,i}\leq-\bar{\mathcal{T}}/2 \right\}\leq e^{-\frac{a^2\vert\mathcal{I}\vert}{32K\bar{a}^2}}.
\end{align}
For the second term, conditioned on $s_{k,i+M},k\in\mathcal{K}$ and $i\in\mathcal{I}$, the variables $z_{i+M}s_{k,i+M}$ are i.i.d. Gaussian variables with zero mean and variance $\sigma^2$. Moreover, $z_{i+M}s_{k,i+M}$ are independent across $i$. Therefore, by symmetry,
\begin{align}
\mathsf{P}\left\{\sum_{i\in\mathcal{I}}z_{i+M}s_{k,i+M}\leq-\bar{\mathcal{T}}/2\right\}=&\mathsf{P}\left\{\sum_{i\in\mathcal{I}}z_{i+M}s_{k,i+M}\geq\bar{\mathcal{T}}/2\right\}\\
=&Q\left(\frac{\bar{\mathcal{T}}}{2\sqrt{\vert\mathcal{I}\vert\sigma^2}}\right)\\
\label{eq:ub_noise_2} \leq&e^{{\color{blue}-}\frac{a^2 |\mathcal{I}| }{32 \sigma^2}}.
\end{align}         

Choosing $\LL{\beta_3} \geq 64 \bar{a}^2 / a^2$, \LL{\eqref{eq:ub_cross_terms_2} and \eqref{eq:ub_noise_2} are both upper bounded by $1/K^2$. Hence, by \eqref{eq:bound_tm_m0},} we have 
\begin{align}
\mathsf{P}\{\text{Re}\{\mathcal{T}(0)\} \leq \frac{2}{K^2}.
\end{align}

Hence the proof of Lemma~\ref{lemma:delay_est}.

\bibliographystyle{IEEEtran}
\bibliography{IEEEabrv,all_it_bib,xu_bib}

\end{document}